\documentclass[10pt, a4paper, twosides]{amsart}
\usepackage{defs_sk}

\title[Discrete time optimal control with frequency constraints]{Discrete time optimal control with frequency constraints for non-smooth systems}
\author[S. Kotpalliwar]{Shruti Kotpalliwar}
\author[P. Paruchuri]{Pradyumna Paruchuri}
\author[D. Chatterjee]{Debasish Chatterjee}
\author[R. Banavar]{Ravi Banavar}
\thanks{The authors were supported partially by the grant 17ISROC001 from the Indian Space Research Organization. They thank Soumitro Banerjee for useful pointers to the literature on control of power electronic circuits, and Karmvir Singh Phogat for helpful discussions and access to his software.}

\date{\today}

\begin{document}
	\maketitle
	\begin{abstract}
		{ We present a Pontryagin maximum principle for discrete time optimal control problems with (a) pointwise constraints on the control actions and the states, (b) smooth \revision{frequency} constraints on the control and the state trajectories, and (c) nonsmooth dynamical systems. Pointwise constraints on the states and the control actions represent desired and/or physical limitations on the states and the control values; such constraints are important and are widely present in the optimal control literature. Constraints of the type (b), while less standard in the literature, effectively serve the purpose of describing important properties of inertial actuators and systems. The conjunction of constraints of the type (a) and (b) is a relatively new phenomenon in optimal control but are important for the synthesis control trajectories with a high degree of fidelity. The maximum principle established here provides first order necessary conditions for optimality that serve as a starting point for the synthesis of control trajectories corresponding to a large class of constrained motion planning problems that have high accuracy in a computationally tractable fashion. Moreover, the ability to handle a reasonably large class of nonsmooth dynamical systems that arise in practice ensures broad applicability of our theory, and we include several illustrations of our results on standard problems.}
	\end{abstract}	
			
	\section{Introduction}
	\label{sec:intro}
		%Optimal control theory deals with the search of a control action for a given dynamical system, which optimizes the predefined cost/performance index. However, from the practical implementation prospective and safety, the designed optimal control and the corresponding state trajectory should respect certain limitations. For example limitations on actuators/thrusters, restrictions on positions and velocity of a systems should be taken into account while designing a control. These limitations are modelled via constraints on the states and control actions in an optimal control problem. Depending on the time domain of system; optimal control problems can be broadly classified as continuous and discrete optimal control problems. In this article we restrict our attention to discrete time optimal control problem.

Optimal control theory, arguably, started in the avatar of the Brachystochrone problem of J.\ Bernoulli in the late 17th century \cite{ref:SusWil-97}, and over the intervening centuries has evolved into a subject that offers a powerful set of tools for control synthesis. Especially relevant in the context of practical applications are synthesis techniques that seamlessly integrate constraints on the states and control actions while maintaining computational tractability. The literature on such constrained problems is certainly not as vast as the counterpart for unconstrained problems, and our article contributes to precisely this body of work.

Under the overarching stipulation of computationally tractable control synthesis techniques, there are two specific contributions of our work. The first concerns the \emph{simultaneous} inclusion of the following four different classes of constraints in control problems:
\begin{itemize}[label=\(\circ\), leftmargin=*]
	\item constraints on the control actions pointwise in time,
	\item constraints on the states pointwise in time,
	\item \revision{frequency} constraints on the control trajectories, and
	\item \revision{frequency} constraints on the state trajectories.
\end{itemize}

\emph{Each} of these four types of constraints is extremely important in practical applications. Actuators are governed by the laws of physics and cannot deliver signals with magnitudes that are beyond their physical limitations; this naturally means that the set of admissible control actions is constrained. Controlled dynamical systems are typically permitted to operate only within reasonable boundaries to avoid fatigue, premature ageing and disintegration of their components; this naturally imposes constraints on their states. Frequency constraints on the control trajectories are needed for all inertial actuators, without which the synthesized control signals may contain frequencies that cannot be faithfully reproduced by the actuators due to their physical limitations; in such cases, differences arise between the predicted and observed behaviours, leading to loss of precision in their desired performances. To take care of this issue, frequently in practice, the synthesized controls are passed through a filter located before the actuator to ensure satisfaction of the desired spectrum. This procedure, however, distorts the original signal that was designed with the 
desired performance objectives in mind, and hence the system performance deteriorates. Frequency constraints on the state trajectories are relevant in inertial controlled systems: they are especially useful to prevent undesirable vibrations of flexible structures in mechanical objects such as satellites, aeroplanes, robotic arms, or to induce desirable damping in flexible structures, etc., and vibration control has been an active area of control for the past several decades.

On the one hand, the first two types of constraints have been studied extensively in the context of optimal control, and are, to a fair extent, amenable to computationally constructive synthesis techniques. Indeed, control action constraints are common place in optimal control theory: both the Pontryagin maximum principle and dynamic programming techniques \cite{Ioffe2009,liberzon2011, ref:Dub-78, ref:DubMil-65, Bertsekas} permit the inclusion of control action constraints pointwise in time. While dynamic programming is not always computationally tractable (e.g., for high-dimensional systems), there are numerical algorithms that employ the Pontryagin maximum principle to synthesize constrained optimal control trajectories. Algorithmic techniques relying on viability theory \cite{aubin2011viability}, while computationally demanding (to the point of being prohibitive for high-dimensional systems), permit the inclusion of state constraints pointwise in time in addition to control constraints. On the other hand, the inclusion of \revision{frequency} constraints on the control and state trajectories is relatively uncommon in the literature. Despite the fact that the classical feedback techniques of continuous-time \(H_\infty\) control \cite{Doyle92} do deal with \revision{frequency domain} behaviour of the control trajectories and are capable of indirectly realizing restrictions \revision{on frequency components} via penalization of certain bands, neither sharp cut-offs in the spectrum nor control and state constraints can be readily enforced via this technique. \revision{Frequency} constraints on the state trajectories, especially in the context of nonlinear systems, have been treated sparsely in the literature; typical existing approaches rely on frequency domain techniques that are ill-suited to handle state and control constraints and are difficult to apply to nonlinear systems. In fact, apart from the US Patent \cite{Brockett}, we are unaware of extensive studies that impose \revision{frequency} constraints on the states and control trajectories. In particular, our own previous studies \cite{ref:ParCha-19, FreqPMPLieGroup} have only focussed on \revision{frequency} constraints on the control trajectories, and the simultaneous inclusion of time and \revision{frequency} constraints with precise bounds and cut-offs on frequency bands are not to be found in the existing literature to the best of our knowledge. Our current results address these lacunae: we establish a version of the Pontryagin maximum principle for optimal control of discrete time nonlinear control systems for which the aforementioned constraints are all enforced \emph{simultaneously}. These constraints are specifically illustrated in \secref{Ex: Inverted Pendulum} on a linear approximation of an inverted pendulum on a cart around the unstable equilibrium position of the pendulum. We specifically impose state constraints so that the validity of the linearization is maintained, and demonstrate that simple point-to-point manoeuvres are dramatically altered by an ad-hoc posteriori application of frequency filters.

The second specific contribution of this article is the ability to encompass nonsmooth problem data. Nonsmooth dynamical systems arise naturally in a variety of application areas \cite{Brogliato}, including those where the physical characteristics of devices involve nonlinearities such as the norm instead of its square \cite{ref:GruPan-17}, or when the dynamical equations are defined piece-wise over disjoint domains of the state-space as a result of intrinsic properties of devices \cite{Banarjee2000}. The Pontryagin maximum principle established in this article, although not catering to the most general situation insofar as nonsmoothness is concerned, handles a large class of nonsmooth nonlinearities that arise in practical applications.\footnote{See also the discussion immediately preceding \eqref{e:DTOCPNSD} concerning possible generalizations.} This specific aspect of our results is illustrated on two examples in \secref{sec:Numerics}: a standard non-smooth dynamical system and the popular buck converter power electronic circuit.

The optimal control synthesis procedure presented in this article proceeds with models in discrete time. We presume that these discrete time models faithfully represent the dynamics of the system, and do not concern ourselves here with the issues of discretization of continuous-time systems. Our preference for discrete time models is motivated by two factors, the first of which is the observation that the final implementation of controllers are carried out digitally via sample-and-hold mechanisms, and therefore, naturally involves a time discretization. The other factor is essentially that of convenience: the mathematical and numerical complexity involved in incorporating the four types of constraints that we want in the continuous time is formidable if not outright impossible.

This article unfolds as follows: We provide in \secref{sec:problem} a precise mathematical statement of our optimal problem that incorporates pointwise constraints on the states and the control actions, and \revision{frequency} constraints on the control and the state trajectories, for a large class of nonsmooth dynamical systems. \secref{sec:prelims} consist of a set of preliminaries needed for the proof of main result. In \secref{sec:main result} we establish a set of first order necessary conditions for optimality in the optimal control problem defined in \secref{sec:problem}; this is the central contribution of this article. We follow up with some technical remarks and immediate corollaries. \secref{sec:proof} is devoted for the detailed proof of main result, and in the final \secref{sec:Numerics} we present several examples to illustrate the proposed necessary conditions. The appendices \secref{Appendix A} and \secref{sec:Lemmas} consist of several auxiliary results that are employed in the proof of our main result.

	\section{Problem Setup}
	\label{sec:problem}
		Consider a discrete time control system whose dynamics is governed by the difference equation
\begin{equation}
\label{eq:sys_dynamics}
	\st[t+1] = \dynamics[t] (\st, \ctrl) \quad \text{for } t = 0, \ldots, \hor-1,
\end{equation}
with the following data:
\begin{enumerate}[label={\textup{(\ref{eq:sys_dynamics}-\alph*)}}, leftmargin=*, widest=b, align=left]
	\item \label{eq:sys:st} \(\st \in \R^{\dimst}\) is the vector of state at time \(t\),
	\item \label{eq:sys:ctrl} \(\ctrl \in \R^{\dimctrl}\) is the vector of control action at time \(t\),
	\item \label{eq:sys:dynamics} \(\R^{\dimst} \times \R^{\dimctrl} \ni (\dummyst[], \dummyctrl[]) \mapsto \dynamics (\dummyst[], \dummyctrl[]) \in \R ^{\dimst}\) for \(t = 0, \ldots, \hor-1\), is a given family of locally Lipschitz continuous maps.
\end{enumerate}

In this article we derive first order necessary conditions \textemdash a Pontryagin Maximum Principle\textemdash for discrete time optimal control problems with nonsmooth dynamics and frequency constraints on the state and control trajectories in addition to the standard  constraints on the control magnitudes and the states. The constraints on the frequencies of the states and control trajectories refer to constraints on the discrete Fourier transform(DFTs) of the individual components of the states and controls. Constraints on the frequency spectra of the control trajectories appeared in the \cite{ref:ParCha-19}; here we move one step further by permitting constraints on the spectra of the state trajectories, to be present as a given stipulation. 

In addition to handling frequency constraints, our result also encompasses piecewise smooth dynamical system such as the ones mentioned in \secref{sec:intro}, and accordingly, we permit our dynamics to be nonsmooth. A more general framework where both the cost-per-stage and the dynamics are nonsmooth, is possible, and such a general framework is typically employed to compare the necessary conditions for a discrete time optimal control problem and a discrete time approximation of a continuous time problem; see, e.g., \cite{Mordukhovich-volII}. We do not aim for maximal generality in our work primarily to ensure a clean calculus; in fact, some of our assumptions are aimed at simplifying the necessary conditions that arise in more general situations.

Here is the precise mathematical statement of our problem:
\begin{equation}
	\boxed{\label{e:DTOCPNSD}
	\begin{aligned}
		& \minimize_{\seq{\ctrl}{t}{0}{\hor-1}} && \sum_{t=0}^{\hor-1} \cost(\st, \ctrl)+\cost[\hor](\st[\hor])\\
		& \sbjto &&
		\begin{cases}
			\text{dynamics } \eqref{eq:sys_dynamics}, \\
			\st \in \admst \quad \text{for } t = 0, \ldots, \hor,\\
			\ctrl \in \admctrl \quad \text{for } t = 0, \ldots, \hor-1,\\
			\stconstraints (\st[0], \ldots, \st[\hor]) = 0, \\
			\ctrlconstraints (\ctrl[0], \ldots, \ctrl[T-1]) = 0, 
		\end{cases}
	\end{aligned}
	}
\end{equation}
where 
\begin{enumerate}[label={\textup{(\ref{e:DTOCPNSD}-\alph*)}}, leftmargin=*, align=left]
	\item \label{eq:sys:cost} the mapping \(\R^{\dimst} \times \R^\dimctrl \ni (\dummyst[], \dummyctrl[]) \mapsto \cost (\dummyst[], \dummyctrl[]) \in \R\) for each \(t = 0, \ldots, \hor-1\), is continuously differentiable and defines a sequence of cost per stage functions, and \(\R^{\dimst}  \ni \dummyst[] \mapsto \cost[\hor] (\dummyst[] ) \in \R\) is a continuously differentiable final stage cost,
	\item \label{eq:sys:admstset} \(\seq{\admst}{t}{0}{\hor}\) is a sequence of closed subsets of \(\R^{\dimst}\) describing a tube (over time) of admissible states,
	\item \label{eq:sys:admctrlset} \(\seq{\admctrl}{t}{0}{\hor-1}\) is a sequence of closed subsets of \(\R^{\dimctrl}\) depicting the sets of admissible control actions at each time \(t\),

	\item \label{eq:sys:SFC} given \(\dimstconstraints \in \N\), some vector \(\Upsilon_{\st[]} \in \R^{\dimstconstraints}\), and smooth functions \(\sfc : \R^{ \dimst} \to \R^{\dimstconstraints}\) for \( t=0, \ldots, \hor \), we define the map \(\R^{\dimst (\hor+1)} \ni (\dummyst[0], \ldots, \dummyst[\hor]) \mapsto \stconstraints (\dummyst[0], \ldots, \dummyst[\hor]) \Let \Upsilon_{\st[]} + \sum_{t=0}^{\hor} \sfc (\dummyst) \in \R^{\dimstconstraints}\) (this map \(\stconstraints\) will be employed to quantify constraints on the frequency spectra of the state trajectories),
%	\item \label{eq:sys:SFC} given \(\dimstconstraints \in \N \), for some vector \( \Upsilon_{\st[]} \in \R^{\dimstconstraints} \) and matrices \(\sfc \in \R^{\dimstconstraints \times \dimst}  \) for all \( t=0, \ldots, \hor \), we define \(  \stconstraints \in \R^{\dimstconstraints \times \dimst (\hor+1) }\) such that  \( \stconstraints (\dummyst[0], \ldots, \dummyst[\hor])\transpose \Let \Upsilon_{\st[]} + \sum_{t=0}^{\hor} \sfc \dummyst  \); this matrix \(\stconstraints\) is employed to quantify constraints on the frequency spectra of the state trajectory. 
	\item \label{eq:sys:CFC} given \(\dimctrlconstraints \in \N\), some vector \(\Upsilon_{\ctrl[]} \in \R^{\dimctrlconstraints}\), and smooth functions \(\cfc :\R^{\dimctrl} \to \R^{\dimctrlconstraints}  \) for \(t=0,\ldots, \hor-1\), we define the map \(\R^{\dimctrl \hor} \ni (\dummyctrl[0], \ldots, \dummyctrl[\hor-1]) \mapsto \ctrlconstraints(\dummyctrl[0], \ldots, \dummyctrl[\hor-1]) \Let \Upsilon_{\ctrl[]} + \sum_{t=0}^{\hor-1} \cfc (\dummyctrl) \in \R^{\dimctrlconstraints}\) (this map \(\ctrlconstraints\) will be employed to quantify constraints on the frequency spectra of the control trajectories).
%	\item \label{eq:sys:CFC} given \(\dimctrlconstraints \in \N \), for some vector \( \Upsilon_{\ctrl[]} \in \R^{\dimctrlconstraints} \) and matrices \(\cfc \in \R^{\dimctrlconstraints \times \dimctrl}  \) for all \( t=0, \ldots, \hor-1 \), we define \( \ctrlconstraints \in \R^{\dimctrlconstraints \times \dimctrl \hor }\) such that  \( \ctrlconstraints(\dummyctrl[0], \ldots, \dummyctrl[\hor-1])\transpose \Let \Upsilon_{\ctrl[]} + \sum_{t=0}^{\hor-1} \cfc \dummyctrl  \); this matrix \(\ctrlconstraints\) is employed to quantify constraints on the frequency spectra of the control trajectory. 
\end{enumerate}

\revision{%
\begin{remark}
	For the control trajectory \(\seq{\ctrl}{t}{0}{\hor-1}\), let \(\ctrl[] \kth \Let (\ctrl[t] \kth)_{t=0}^{\hor-1}\) denote the trajectory of the \(k^{\text{th}}\) component of the control for each \(k = 1, \ldots, \dimctrl\). For each component trajectory \(\ctrl[] \kth\), we obtain the vector of DFT coefficients \(\widehat{\ctrl[] \kth} \Let \mathbf{F} \ctrl[] \kth \in \mathbb{C}^{\hor}\) where \(\mathbf{F} \in \mathbb{C}^{\hor \times \hor}\) denotes the DFT matrix corresponding to signal length \(\hor\) (see \cite[Chapter 7]{Stein-Shakarchi}). Let \(\widehat{\ctrl[]} \Let \pmat{\widehat{\ctrl[] \kth[1]} \transpose & \widehat{\ctrl[] \kth[2]} \transpose & \ldots & \widehat{\ctrl[] \kth[\dimctrl]} \transpose} \transpose \in \mathbb{C}^{\dimctrl \hor}\) denote the vector of stacked frequency coeffiecients of all components \(\ctrl[] \kth\), \(k = 1, \ldots, \dimctrl\). That is,
	\[
	\mathbb{C}^{\dimctrl \hor} \ni \widehat{\ctrl[]} \Let \pmat{\widehat{\ctrl[] \kth[1]}\\ \vdots\\ \widehat{\ctrl[] \kth[\dimctrl]}} = \pmat{\mathbf{F} \ctrl[] \kth[1]\\ \vdots\\ \mathbf{F} \ctrl[] \kth[\dimctrl]}
	\]
	By frequency constraints on the control trajectory, we refer to constraints on the DFT coefficients of the trajectories of the components of control. We specialize our results in this article to constraints that eliminate certain frequency components of the control trajectories, i.e., restrict certain entries of \(\widehat{\ctrl[]}\) to zero. By carefully choosing the rows corresponding to the desired frequency components to be restricted, separating the real and imaginary parts, and using a suitable rearrangement, one obtains constraints of the form mentioned in \ref{eq:sys:CFC}, with \(\dimctrlconstraints\) denoting the number of entries of \(\widehat{\ctrl[]}\) being restricted, counting real and imaginary parts separately, and the maps \(\cfc\) being the linear transformations containing the rows of DFT matrix in a specific arrangement. In our presentation, we generalize this class of constraint maps,  \(\cfc\), to be just smooth. For more detailed explanation see \cite{arXivCopy}.
	
	The frequency constraints on the state trajectories described in \ref{eq:sys:SFC} are obtained in a similar manner. If one were to consider the DFT coefficients, one needs to use the DFT matrix corresponding to a signal of length \(\hor + 1\).
\end{remark}}

In the sequel we refer to the sequence \( \seq{\ctrl}{t}{0}{\hor-1} \) as a \emph{control trajectory}, with \( \ctrl\) being the \emph{control action} at time \(t\).

	\section{Preliminaries}
	\label{sec:prelims}
		In this section we review some elementary facts from analysis in a nonsmooth setting; we refer the reader to \cite{Clarke, Baptiste, Guler} for further information on this topic.

For us \(\inprod{\cdot}{\cdot}\) refers to the usual inner product in \(\R^{\genDim}\) and \(\norm{\cdot}\) to the norm induced by \(\inprod{\cdot}{\cdot}\). For \(\epsilon > 0\) and \(\genVar \in \R^{\genDim}\), \(\ball{\genVar}\) refers to the open ball \(\set{y \in \R^{\genDim} \suchthat \norm{y - \genVar} < \epsilon}\) of radius \(\epsilon\) centered at \(\genVar\).

 A function \(\function \colon \R^{\genDim} \lra \R\) is said to be \emph{Lipschitz continuous near \(\genVar \in \R^{\genDim}\)} if there exists \( r,L > 0\) such that
\[
	\abs{\function(y)-\function(z)}\leq L \norm{y-z}\quad \text{for all } y, z \in \ball[r]{\genVar}.
\]
	A set \(\genSet \subset \R^{\genDim}\) is \embf{convex} if for all \(x, y \in \genSet\) and  for all \(\alpha \in [0,1]\), we have \((1 - \alpha) x +  \alpha y \in \genSet\). A nonempty  set \( K \subset \R^\genDim\) is a \embf{cone} if for every \(y \in K\) and for every  \(\alpha \geq 0, \) we have \( \alpha y \in K \). A set \(\ccone \subset \R^{\genDim}\) is a \embf{cone with vertex} \(\vertex \in \R^{\genDim}\) if  for all  \(\alpha \ge 0\) and for all \(y \in \ccone\), \(\alpha (y - \vertex) \in \ccone\). For a convex cone \(\ccone \subset \R^{\genDim}\) with vertex \(\vertex\), its \embf{dual cone} \(\dualCone{\ccone}\) is defined via polarity by
	\[
		\dualCone{\ccone} \Let \set[\big]{y \in \R^{\genDim} \suchthat \inprod{y}{\genVar - \vertex} \le 0 \quad \text{for all } \genVar \in \ccone}.
	\]
\begin{definition}\label{def:directional derivative}\cite[Section 1.4 on p.\ 20]{Clarke-13}
	Let \(\function \colon \R^{\dimst} \times \R^{\dimctrl} \lra \R^{\genDim}\) be a continuous map and \(\function_{1}, \ldots, \function_{\genDim}\) are its components. For \(y \in \R^{\dimctrl}\) and a vector \(\dir \in \R^{\dimst}\), we denote by \(\dirder{\function(\cdot, y)}(x)\) \embf{the directional derivative of \(\function(\cdot, y)\)} along \(\dir\) at \(x\), whenever the following limit exists:
	% \[
	% 	\R^{\dimst} \times \R^{\dimst} \ni (\st[], \dir) \mapsto \dirder{\function(\cdot, y)}(\st[]) \in \R^{\genDim}
	% \]
	% where,
	\begin{align}\label{eq:directional derivative}
	 \dirder{\function(\cdot,y)}(\st[]) \Let & \lim\limits_{\theta \downarrow 0} \frac{\function(\st[] + \theta \dir, y) - \function(\st[], y)}{\theta} \\
		= & \left(\lim \limits_{\theta \downarrow 0} \frac{\function_1(\st[] + \theta \dir, y) - \function_1(\st[], y)}{\theta}, \ldots, \lim\limits_{\theta \downarrow 0}\frac{\function_{\genDim}(\st[] + \theta \dir, y) - \function_{\genDim}(\st[], y)}{\theta}\right)\transpose \nonumber
	\end{align}
	We note that the directional derivative above is defined as a right-hand (one-sided) limit. If \( \function \) is continuously differentiable, then \( \dirder{\function(\cdot,y)}(\st[]) = \derivative{\function}{\st[]}(\st[],y). \dir\). However, there are Lipschitz functions for which the left-hand and right-hand limit exist for all the directions at a point \( \genVar\) (which happens if the function is continuous at \genVar), but they are not equal i.e.,  \(\lim\limits_{\theta \downarrow 0} \frac{\function(\st[] + \theta \dir) - \function(\st[])}{\theta}\neq\lim\limits_{\theta \uparrow 0} \frac{\function(\st[] + \theta \dir) - \function(\st[])}{\theta}\).
	For example, for the absolute value function \(\R\ni \genVar \mapsto f(\genVar)=\abs{\genVar} \in \R \) at \( \genVar =0\), \(\dirder{f(0)}=\abs{v}=\begin{cases}
	\dir & \text{ for } \dir \geq 0, \\
	-\dir & \text{ for } \dir < 0,
	\end{cases}\) and \(\lim\limits_{\theta \uparrow 0}\frac{f(0+\theta \dir)-f(0)}{\theta}= -\abs{v}.\) Of course, \(\abs{\cdot} \) is not Gateaux differentiable at \(\genVar=0\).\footnote{Recall that if the right-hand limit of \(\function \colon \R^{\dimst} \times \R^{\dimctrl} \lra \R^{\genDim}\) at \(\genVar\) is equal to the left-hand limit at  \(\genVar\) then, the function \( \function\) is \emph{Gateaux differentiable} at  \(\genVar\). In this case, \(\dirder{\function(\cdot,y)}(\st[])=\derivative{\function(\cdot,y)}{\dummyst[]}(x) \cdot \dir\).}

	\end{definition}
	\begin{definition}
	If \(h : \R^{\genDim} \lra \R\) is a Lipschitz continuous function, then the \embf{generalized directional derivative} \(\gdd[h]{\genVar}{\dir}\) of \(h\) at \(\genVar\) along the direction \(\dir\) is defined by
	\[
		\gdd[h]{\genVar}{\dir} \Let \gdddef{y}{\genVar}{t}{0}{h}{\dir}.
	\]
	In general, the generalized directional derivative takes values in \(\R \cup \{+ \infty\}\).
	The \embf{generalized gradient} \(\ggrad{h}{\genVar}\) of \(h\) at \(\genVar\) is a nonempty compact subset of \(\dualspace{\R^{\genDim}}\) defined by
	\[
		\ggrad{h}{\genVar} \Let \set[\big]{\dummyst[] \in \dualspace{\R^{\genDim}} \suchthat \inprod{\dummyst[]}{\dir} \le \gdd[h]{\genVar}{\dir} \quad \text{for all } \dir \in \R^{\genDim}}.
	\]
\end{definition}

\revision{\begin{definition}
	A Lipschitz continuous function \(\function \colon \R^{\dimst} \lra \R\) is said to be \embf{regular} at \( x \in \R^{\dimst} \) if the directional derivative of \( \function \) at \(x\) along any \( \dir \in \R^{\dimst} \) exists and is equal to its generalized directional derivative at \( x \) along  that \( \dir \), i.e., 
		\[
			\dirder{\function}(x) = \gdd[\function]{x}{\dir} < + \infty.
		\]
\end{definition}}

We look at two examples:
	\begin{itemize}[leftmargin=*, label=\( \circ \)]
		\item Let \( \R \ni x \mapsto f(x)\Let \text{max} \{0, x\} \in \R\). It is clear that \(f\) is Lipschitz continuous and  is differentiable everywhere except at \(x=0\). 
		The directional derivative, the  generalized directional derivative and the generalized gradient of \( f \) at \(0\) and \(1\) are
		\[
		 \revision{\dirder{f}(0) = f(\dir),} \quad \gdd[f]{0}{\dir}=f(\dir), \quad  \ggrad{f}{0}=[0,1], 
		\]
		\[
		  \revision{ \dirder{f}(1) = \dir,} \quad \gdd[f]{1}{\dir}=\dir \quad \text{ and } \quad \ggrad{f}{1}=\{1\}.
		\] 
		\item Let \( \R^{\dimst} \ni x \mapsto f(x) \Let \norm{x} \in \R\). Clearly, \(f\) is a Lipschitz continuous function differentiable everywhere except at \(0\). 
			The  generalized directional derivative and the generalized gradient of \( f \) at \(0\) are 
			\[
			 \revision{\dirder{f}{(0)} = \norm{\dir},} \quad \gdd[f]{0}{\dir}=\norm{\dir} \text{ and } \quad  \ggrad{f}{0}=\closure (\ball[1]{0}).
			\] 
	\end{itemize}
	\revision{Note that the functions in the above two examples are regular at \(0 \).}

	Let \(\genSet \subset \R^{\genDim}\) be a nonempty and closed set. The \embf{distance} \(\dist{\genSet}(\genVar)\) of a point \(\genVar \in \R^{\genDim}\) from \(\genSet\) is defined by 
	\[
		\R^{\genDim} \ni \genVar \mapsto \dist{\genSet}(\genVar) \Let \inf_{s \in \genSet} \norm{\genVar - s} \in [0, +\infty[.
	\]
\begin{definition}	
	The \embf{Clarke tangent and normal cones} to  \(\genSet\) at a point \(\genVar \in \genSet\), denoted by \(\Tcs_{\genSet}(\genVar)\) and \(\Ncs_{\genSet}(\genVar)\) respectively, are defined as
	\begin{align*}
		& \Tcs_{\genSet}(\genVar) \Let \set[\big]{\dir \in \R^{\genDim} \suchthat \gdd[\dist{\genSet}]{\genVar}{\dir} \le 0},\\
		& \Ncs_{\genSet}(\genVar) \Let \set[\big]{\dummyst[] \in \dualspace{\R^{\genDim}} \suchthat \inprod{\dummyst[]}{\dir} \le 0 \quad \text{for all } \dir \in \Tcs_{\genSet}(\genVar)}.
	\end{align*}
\end{definition}

\begin{figure}[!h]
	\begin{tikzpicture}[
    scale=2.5,
    axis/.style={very thick, ->, >=stealth'},
    important line/.style={thick, color= blue},
    dashed line/.style={dashed, thin},
    pile/.style={thick, ->, >=stealth', shorten <=2pt, shorten
    >=2pt},
    every node/.style={color=black}
    ]
    \begin{scope}
    \coordinate (O) at (0,0);
      \coordinate (A) at (-1,1);
      \coordinate (C) at (1,1);
    % axis
    \draw[axis] (-1,0)  -- (1.1,0) node(xline)[right]
        {$x_{1}$};
    \draw[axis] (0,-1) -- (0,1.3) node(yline)[above] {$x_{2}$};
    % Lines
   \draw[important line] (O) -- (A)
        node[right, text width=5em] {$x_{2}=|x_{1}|$};
   \draw[important line] (O)  -- (C);
    %Shade
    \shade[left color=blue, bottom color=blue, right color=blue, opacity=0.25]
   		(-1,1) -- (-0.9,1.0618) -- (-0.8,1.1176) -- (-0.7,1.1618)-- (-0.6, 1.1902) -- (-0.5,1.2) -- 
   		(-.4,1.1902) -- (-.3, 1.1618) -- (-.2, 1.1176) -- (-0.1,1.0618 )  -- (0,1) -- (0.1, 0.9382) --
   		(0.2, 0.8824) -- (0.3, 0.8382 ) -- (0.4, 0.8098 ) -- (0.5, 0.8 ) -- (0.6, 0.8098) -- (0.7, 0.8382) --
   		(0.8, 0.8824) -- (0.9, 0.9382) -- (1,1) -- (0,0) -- cycle;
   	%Shade
    \shade[left color=red, bottom color=red, right color=red, opacity=0.35]
         (1,1) -- (0.9,1.0618) -- (0.8,1.1176) -- (0.7,1.1618)-- (0.6, 1.1902) -- (0.5,1.2) -- 
        (.4,1.1902) -- (.3, 1.1618) -- (0.2, 1.1176) -- (0.1,1.0618 )  -- (0,1) -- (-0.1, 0.9382) --
        (-0.2, 0.8824) -- (-0.3, 0.8382 ) -- (- 0.4, 0.8098 ) -- (- 0.5, 0.8 ) -- (- 0.6, 0.8098) -- 
        (-0.7, 0.8382) -- (-0.8, 0.8824) -- (-0.9, 0.9382) -- (-1,1) -- (0,0) -- cycle;                 
    %tangent cone
     %\draw[very thick, dashed, color=red,  fill=red, opacity=0.25]
      %   plot[smooth]
        % coordinates{(-0.84,0.84) (-0.5,.88) (-0.2,0.82) (0,0.8) (0.2,0.82) (0.5,.8) (0.83,0.85) (0.8,0.8) (0.7,0.7) %(0.6,0.6) (0,0)}; 
    %coordinates{
     %(-0.8,0.8) (-0.7, 0.9)  (-0.5,.98) (-0.3,1) (-0.1,0.98) (0,0.95) (0.2,0.9) (0.4,.88) (0.6,0.85) (0.8,0.8) %(0.809,.809) (0,0)}; 
    %Normal Cone
    \draw[very thick, color=orange] (O) -- (-1,-1);
    \draw[very thick, color=orange] (O) -- (1,-1);
    \shade[left color=orange, bottom color=orange, right color=orange, opacity=0.25]
                  (-1,-1) -- (0,0) -- (1,-1)-- cycle;
   %Nomenclature
    \node[above] at (0,0) [anchor=north east] {$o$};
   \node[ color= blue, above] at (0.25,0.81){ $S_1$};
   \node[ color=red, above] at (0.25,0.4) {$T^{C}_{S_{1}}(o)$};   
   \node[color=orange, above] at (0.25,-0.8) {$N^{C}_{S_{1}}(o)$} ;       
   \end{scope}
%    % % % % % % % %
    \begin{scope}[shift={(3,0)}]
    %coordinates
    	\coordinate (O) at (0,0);
        \coordinate (A) at (0.5,1);
        \coordinate (B) at (-0.5,1);
    % axis
       \draw[axis] (-1,0)  -- (1.1,0) node(xline)[right]{$x_{1}$};
       \draw[axis] (0,-1) -- (0,1.1) node(yline)[above] {$x_{2}$};
    % set
    \draw[important line] (O) -- (A);  
    \draw[important line] (O) -- (B) ;
    \node[above] at (0.6,.5) {$x_{2}=2|x_{1}|$};
    \shade[left color=blue, bottom color=blue, right color=blue, opacity=0.25]
    	(-0.5,1) -- (-1,1) -- (-1,-1) -- (1,-1) -- (1,1) -- (0.5,1) -- (0,0)-- cycle; 
    
    %Tangent cone
    \draw[very thick, color=red] (O) -- (-0.48,-0.96);
    \draw[very thick, color=red] (O) -- (0.48,-0.96);
    \draw[very thick, dashed, color=red,  fill=red, opacity=0.25]
    	plot[smooth] 
        coordinates{(-0.48,-0.95)  (-0.2,-0.93) (0,-0.88) (0.2,-0.9) (0.48,-0.97) (0.38,-0.76) (0,0) };   
   %Normal Cone
   \draw[very thick, color=orange] (O) -- (-0.9,0.45);
   \draw[very thick, color=orange] (O) -- (0.9,0.45);
   % \shade[left color=orange, bottom color=orange, right color=orange, opacity=0.35]
   %                  (-1,0.5) -- (0,0) -- (1,0.5)-- (1,1) -- (-1,1) -- cycle;
   \shade[left color=orange, bottom color=orange, right color=orange, opacity=0.35]
   					(-0.90, 0.4500) -- (-0.85, 0.4718) -- (-0.80, 0.4934) -- (-0.75, 0.5147) -- (-0.70, 0.5355) -- (-0.65, 0.5557) -- (-0.60, 0.5750) -- (-0.55, 0.5934) -- (-0.50, 0.6107) -- (-0.45, 0.6268) -- (-0.40, 0.6415) -- (-0.35, 0.6548) -- (-0.30, 0.6665) -- (-0.25, 0.6766) -- (-0.20, 0.6849) -- (-0.15, 0.6915) -- (-0.10, 0.6962) -- (-0.05, 0.6990) -- (0.00, 0.7000) -- (0.05, 0.6990) -- (0.10, 0.6962) -- (0.15, 0.6915) -- (0.20, 0.6849) -- (0.25, 0.6766) -- (0.30, 0.6665) -- (0.35, 0.6548) -- (0.40, 0.6415) -- (0.45, 0.6268) -- (0.50, 0.6107) -- (0.55, 0.5934) -- (0.60, 0.5750) -- (0.65, 0.5557) -- (0.70, 0.5355) -- (0.75, 0.5147) -- (0.80, 0.4934) -- (0.85, 0.4718) -- (0.90, 0.4500) -- (0, 0) -- cycle;
   %Nomenclature
   \node[color= red, above] at (0.21,-.94){$T_{S_{2}}^{C}(o)$};
   \node[color= blue,  above] at (.5,-0.5){ $S_2$};
   \node[color=orange, above] at (0.21,0.74){$N_{S_{2}}^{C}(o)$};
   \node at (0,0) [anchor=north east] {$o$};
   \end{scope}
   \end{tikzpicture}   
	\caption{Tangent cone to the sets \(S_{1}, S_{2}\) at \(o=(0,0)\)}
	\label{fig:tangentcone1}
\end{figure}
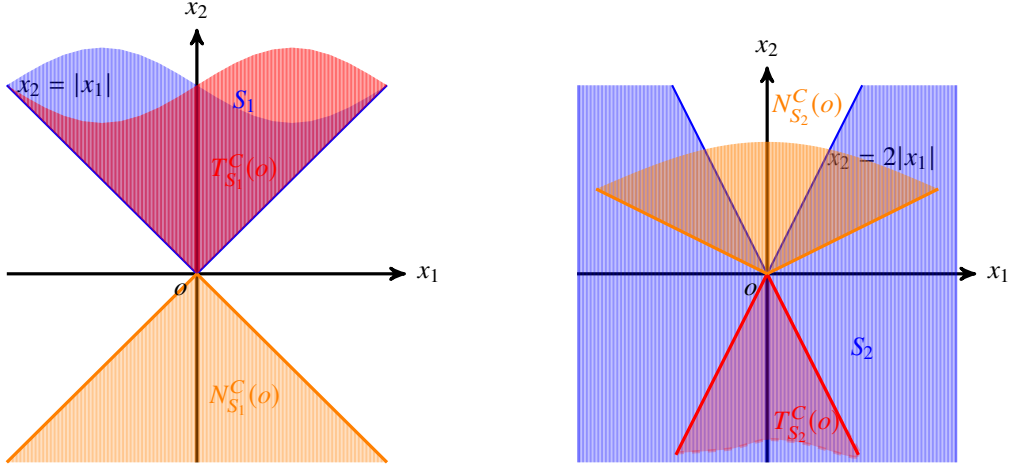
The Clarke tangent cone \(\Tcs_{\genSet}(\genVar)\) at \(\genVar\) is a closed convex set and the normal cone \(\Ncs_{\genSet}(\genVar)\) at \(\genVar\) is the polar of the tangent cone \(\Tcs_{\genSet}(\genVar)\). Intuitively, the tangent cone at \(\genVar\) to \(\genSet\) is the set of directions along which it is possible to `enter' \(\genSet\) from \(\genVar\), and the normal cone provides the set of directions along which one can most efficiently `exit' \(\genSet\) from \(\genVar\). The Clarke tangent cone and normal cone to the sets \(S_{1}=\{x\in\R^2| x_{2} \geq \abs{x_{1}}\}\) and \(S_{2}=\{x\in\R^2| x_{2} \leq 2\abs{x_{1}}\}\) at origin $o=(0,0)$ are shown in Figure \ref{fig:tangentcone1}.

	\section{Main result}
	\label{sec:main result}
		\revision{\begin{assumption}
			\label{assum:regularity of dynamics}
	We stipulate that the maps \( \dynamics \)'s in \ref{eq:sys:dynamics} are regular at every \( (\dummyst[], \dummyctrl[]) \in  \R^{\dimst \times \dimctrl}\).
\end{assumption}
}

Our main result is the following theorem:\footnote{In the sequel there will arise the need to take partial derivatives of multivariable functions relative to specific variables, and we take care to precisely indicate the variables with respect to which we take these partial derivatives by introducing dummy variables in the definitions. The adjoint variables (a.k.a.\ `multipliers') corresponding to the cost, the dynamics, the state-constraints, and the frequency constraints of the state and the control trajectories appear here, and we distinguish between them by introducing the different super-scripts of the single Greek letter \(\eta\). Various objects in frequency space are distinguished by a `hat'. While this mechanism leads to a multitude of sub-/super-scripts, we believe that it ensures much-needed transparency and tractability.}

\begin{theorem}
\label{th:nonsmooth pmp}
	\revision{Consider the problem \eqref{e:DTOCPNSD} along with its associated data, and suppose that Assumption \ref{assum:regularity of dynamics} holds. Let \(\seq{\optCtrl}{t}{0}{\hor-1}\) be a control trajectory that solves the optimal optimal control problem \eqref{e:DTOCPNSD}, and let \(\optState[] \Let \seq{\optState}{t}{0}{\hor}\) be the corresponding state trajectory.} Define the Hamiltonian
	\begin{equation}
	\label{e:hamiltonian}
	\begin{aligned}
		& \R \times \dualspace{\R^{\dimstconstraints}} \times \dualspace{\R^{\dimctrlconstraints}} \times \dualspace{\R^{\dimst}} \times \{0,\ldots,\hor-1\} \times \R^{\dimst} \times \R^\dimctrl \ni \left(\dummymulticost, \dummymultisfc, \dummymulticfc, \dummyadjoint,\dummytime,\dummyst[],\dummyctrl[] \right) \mapsto\\
		& \quad \hamiltonian[\dummymulticost, \dummymultisfc, \dummymulticfc](\dummyadjoint, \dummytime, \dummyst[], \dummyctrl[]) \Let \inprod{\dummyadjoint}{\dynamics[\dummytime] (\dummyst[], \dummyctrl[])} -\dummymulticost \cost[\dummytime] (\dummyst[], \dummyctrl[]) -\inprod{\dummymultisfc}{\sfc[\dummytime] (\dummyst[])} -\inprod{\dummymulticfc}{\cfc[\dummytime] (\dummyctrl[])} \in \R.
	\end{aligned}
	\end{equation}
	Then there exist \begin{itemize}
	\item an adjoint trajectory \(\seq{\adjoint}{t}{0}{\hor-1} \subset \dualspace{\R^\dimst}\),
	\item a sequence \(\seq{\multiplierstNcs}{t}{0}{\hor} \subset \dualspace{\R^\dimst}\), and
	%\seq{\multiplierctrlNcs}{t}{0}{\hor-1} \subset \dualspace{\R^\dimctrl}\), %\(\seq{\multiplierctrlNcs}{t}{0}{\hor} \subset \dualspace{\R^\dimctrl}\),
	\item a triplet \(\bigl(\multipliercost, \multipliersfc, \multipliercfc \bigr) \in \R \times \dualspace{\R^{\dimstconstraints}} \times \dualspace{\R^{\dimctrlconstraints}}\),
	\end{itemize} satisfying
	\begin{enumerate}[label={\textup{(C-\roman*)}}, leftmargin=*, align=left, widest=iii]
		\item \label{non-negativity} the non-negativity condition: \(\multipliercost \in \set{0, 1};\)
	\item \label{non-triviality} the non-triviality condition: the adjoint trajectory \(\seq{\adjoint}{t}{0}{\hor-1}\) and the triplet \(\bigl(\multipliercost, \multipliersfc, \multipliercfc \bigr)\) do not simultaneously vanish;
	\item \label{state and adjoint_equation} the state and adjoint dynamics:
		\begin{align*}
			& \optState[t+1]  =   \derivative{\hamiltonian (\adjoint, t, \optState, \optCtrl)}{\dummyadjoint} \quad \text{for } t = 0, \ldots, \hor-1, \\
			& \inprod{\adjoint[t-1]}{\dirst}  \geq \dirder[\dirst]{\hamiltonian (\adjoint,t,\cdot, \optCtrl)}(\optState)-\inprod{ \multiplierstNcs  }{\dirst}\\
			& \qquad\qquad\qquad\text{for all \(y \in \R^{\dimst}\) and some \( \multiplierstNcs \in \Ncs_{\admst}(\optState)\),  for \( t=1, \ldots, \hor-1\)};
		\end{align*}	  
 	\item \label{transversality} the transversality conditions:
		\begin{align*}
			& \dirder[\dirst]{\hamiltonian (\adjoint[0], 0, \cdot, \optCtrl[0])}(\optState[0]) - \inprod{\multiplierstNcs[0]}{\dirst} \leq 0 \\
			  & \qquad\qquad\qquad\text{ for all } y \in \R^{\dimst} \text{ and some }\multiplierstNcs[0] \in \Ncs_{\admst[0]}(\optState[0]), \text{ and } \\
	   	     & \adjoint[\hor-1] = -\multipliercost \derivative{\cost[\hor]}{\dummyst[]}(\optState[\hor]) -\biggl(\derivative{\sfc[\hor]}{\dummyst[]} (\optState[])\biggr) \transpose \multipliersfc - \multiplierstNcs[\hor]\\
			 & \qquad\qquad\qquad\text{for some \(\multiplierstNcs[\hor] \in \Ncs_{\admst[\hor]}(\optState[\hor])\)};
	   \end{align*}
	\item \label{Hamiltonian Maximization} the Hamiltonian maximization condition:
		\begin{align*}
			\dirder[\dirctrl]{\hamiltonian (\adjoint, t, \optState, \cdot)}(\optCtrl) \leq 0\quad
		 \text{ for all }  \dirctrl \in \Tcs_{\admctrl}(\optCtrl) \;\text{ and for } t=0 \ldots, \hor-1;
	    \end{align*}	  
	\item \label{State Frequency Constraints} frequency constraints on the state trajectory \seq{\optState}{t}{0}{\hor}:
		\[
			\stconstraints(\optState[0], \ldots, \optState[\hor]) = 0;
		\]
	\item \label{Control Frequency Constraints} frequency constraints on the control action trajectory \seq{\optCtrl}{t}{0}{\hor-1}:
		\[
			\ctrlconstraints(\optCtrl[0], \ldots, \optCtrl[\hor-1]) = 0.
		\]
	\end{enumerate}
\end{theorem}
A detailed proof of Theorem \ref{th:nonsmooth pmp} is postponed to the next section. In the remainder of the current section we shall \revision{briefly discuss the conditions \ref{non-negativity}--\ref{Control Frequency Constraints} in the forthcoming remarks and} examine several special cases of the main problem \eqref{e:DTOCPNSD}.
\begin{remark}
	\label{Remark:Abnormal Multiplier}
 	The scalar  \( \multipliercost\), called the abnormal multiplier, takes the value \(0 \) or \(1\).  If \( \multipliercost=0\), then  the \emph{extremal lift}  \( \left( \multipliercost,  \seq{\adjoint}{t}{0}{\hor}, \multipliersfc, \multipliercfc, \seq{\optState}{t}{0}{\hor}, \seq{\optCtrl}{t}{0}{\hor-1} \right) \) corresponding to an optimal pair \( \left(\seq{\optState}{t}{0}{\hor}, \seq{\optCtrl}{t}{0}{\hor-1}\right) \) is called an \emph{abnormal extremal}. If \( \multipliercost=1\), then the corresponding the extremal lift is called \emph{normal}. 
 \end{remark}
 \begin{remark}
 	\label{Remark:adjoint Equation}
  		The entries of the sequence \(\seq{\adjoint}{t}{0}{\hor-1}\) are called adjoint vectors or co-states; their evolution is governed by the adjoint dynamics \ref{state and adjoint_equation}, and the transversality conditions \ref{transversality} provide its boundary conditions. Due to the nonsmooth nature of the state dynamics, the adjoint recursion is an inclusion
		\[
			\adjoint[t-1] \in \set[\big]{\eta \in \R^{\dimst} \suchthat \inprod{\eta}{y} \ge \dirder[\dirst]{\hamiltonian (\adjoint,t,\cdot, \optCtrl)}(\optState)-\inprod{ \multiplierstNcs  }{\dirst} \text{ for all } y \in \R^{\dimst}}
		\]
		as opposed to an equation. In Corollary \ref{cor:smooth sys PMP} we shall observe that if the dynamics are continuously differentiable, then this inclusion turns into an equation, which is the standard adjoint equation in the classical Pontryagin maximum principle.
\end{remark}
\revision{\begin{remark}
In the transversality condition \ref{transversality} the terms \( \Ncs_{\admst[0]}(\optState[0])\) and \(\Ncs_{\admst[\hor]}(\optState[\hor]) \) are Clarke normal cones. However, there is a more general notion of a normal cone due to Mordukhovich called the \emph{basic} or \emph{limiting normal cone} \cite[Chapter 1]{Mordukhovich-volI}, and the approximate discrete-time PMP in \cite{Mordukhovich-volII} presents necessary conditions for optimality in terms of this limiting normal cone. Our approach of the proof uses the convexity of the Clarke normal cone in an essential way, while it is known that the limiting normal cone may fail to be convex; consequently, the result and its proof provided here does not carry over in an elegant fashion while involving the limiting normal cone.
\end{remark}}

\begin{remark}
	\label{Remark:freq Constraints terms}
	The sequence of multipliers \(\seq{\multiplierstNcs}{t}{0}{\hor}\) correspond to the pointwise state constraints.
	The definition \eqref{e:hamiltonian} of the Hamiltonian  features two new terms, \( \inprod{\dummymultisfc}{\sfc[\dummytime] \dummyst[]}\) and \(\inprod{\dummymulticfc}{\cfc[\dummytime] \dummyctrl[]}, \) compared to the standard definition (e.g., \cite{ref:Bol-75}). These terms account for the frequency constraints on the state trajectory and the control trajectory, respectively. One similar term also appeared in the Hamiltonian  in \cite{ref:ParCha-19}, where frequency constraints on the control trajectories were considered.	The vectors \(\multipliersfc, \multipliercfc \) are the multipliers corresponding to the frequency constraints on the states and the controls, respectively.
\end{remark}
\begin{remark}
	\label{Remark:H Maximization}
	  %The condition \ref{Hamiltonian Maximization} states that if \(\dirctrl \in \Tcs_{\admctrl}(\optCtrl)\) is a direction from \(\optCtrl,\) then the directional derivative of the Hamiltonian in the direction \(\dirctrl\) is non-positive. Since the tangent cone contains the directions along which it  is possible to enter the set \(\admctrl\) from \(\optCtrl\), this implies that the Hamiltonian is non-increasing locally along the direction \(\dirctrl\) in \( \Tcs_{\admctrl}(\optCtrl)\). This does not imply that the Hamiltonian is at its maximum at \(\optCtrl\). However, if we assume convexity and compactness of each admissible action set \(\admctrl \), we can claim the maximization of the Hamiltonian at each instant along the optimal trajectories. We shall witness this feature in the greater detail in the upcoming corollary. 
	  \revision{The condition \ref{Hamiltonian Maximization} states that at every \( t \) there exist a neighborhood of \( \optCtrl\) (say \(\ball{\optCtrl}\) ) such that
	  	\[
		  	\hamiltonian (\adjoint, t, \optState, \optCtrl + p) \leq \hamiltonian (\adjoint, t, \optState, \optCtrl ) \quad \text{ for all } p \in \ball{\optCtrl} \cap \Tcs_{\admctrl}(\optCtrl).
	  	\] 
	  The Clarke tangent cone \( \Tcs_{\admctrl}(\optCtrl) \) provides a convex conic approximation of the set \(\admctrl \) at the point \(\optCtrl\). (Here the term approximation of a set \(\admctrl \) at a point \(\optCtrl\) stands for a set of directions along which it is possible to enter the set \(\admctrl \) from the point \( \optCtrl \).) Then the condition \ref{Hamiltonian Maximization} implies that the value Hamiltonian at \(\optCtrl\) is greater than the value Hamiltonian at the points in the convex conic approximation of the set \(\admctrl \) at the point \(\optCtrl\) and that are close to \( \optCtrl \). Consequently, \ref{Hamiltonian Maximization} does not imply Hamiltonian maximization which is a well-known phenomenon in the discrete time optimal control literature \cite{Mor04,pshenichnyi71}. Although not entirely appropriate, we call this condition ``the Hamiltonian maximization'' to maintain a similarity with the continuous time PMP. Besides, under suitable further assumptions on the admissible control action set, \ref{Hamiltonian Maximization} implies maximization of the Hamiltonian. Indeed, if the admissible set \( \admctrl \) is convex and compact, then \(\admctrl \subset \Tcs_{\admctrl}(\optCtrl) \), and the condition \ref{Hamiltonian Maximization} simplifies to
		\begin{align*}
			& \text{at every \(t\) there exist a neighborhood of \( \optCtrl\), i.e., \(\ball{\optCtrl} \), such that}\\
			& \hamiltonian (\adjoint, t, \optState, \optCtrl + p) \leq \hamiltonian (\adjoint, t, \optState, \optCtrl ) \quad \text{ for all } p \in \ball{\optCtrl} \cap \admctrl.
		\end{align*}
	  In other words, \( \optCtrl \) is a local maximizer of the Hamiltonian at time \(t\). Therefore, if we assume that at each \( t \) the Hamiltonian \ref{e:hamiltonian} is concave with respect to the control variable then the condition \ref{Hamiltonian Maximization} implies the local maximization of the Hamiltonian at each instant along the optimal trajectories.}
\end{remark}

\begin{corollary}
	Suppose that \( \bigl( \multipliercost,  \seq{\adjoint}{t}{0}{\hor}, \multipliersfc, \multipliercfc, \seq{\optState}{t}{0}{\hor}, \seq{\optCtrl}{t}{0}{\hor-1} \bigr) \) is an extremal of \eqref{e:DTOCPNSD} and let \( \hamiltonian \) be the Hamiltonian defined in \eqref{e:hamiltonian}.
\begin{enumerate}[label={\textup{(\roman*)}}, leftmargin=*, align=left, widest=iii]
	\item \label{smooth state and adjoint_equation} 
	If the dynamics \( \dynamics[i]\) in \ref{eq:sys:dynamics} for every \( i \in \set{1,\ldots, \hor-1}\) are smooth with respect to the state variable \( \dummyst[] \) at \(\optState[i]\), then the adjoint dynamics in  \ref{state and adjoint_equation} of Theorem \ref{th:nonsmooth pmp} at \( t = i\) strengthens to the recursion
	\begin{align*}
		 \adjoint[i-1]  = \derivative{\hamiltonian (\adjoint[i], i, \optState[i], \optCtrl[i])}{\dummyst[]}- \multiplierstNcs[i]  \quad \text{for some }  \multiplierstNcs[i] \in \Ncs_{\admst[i]}(\optState[i]).
	\end{align*}	
	\item \label{smooth transversality} 
		If the dynamics \(\dynamics[0] \) in \ref{eq:sys:dynamics} is smooth with respect to the state variable \( \dummyst[] \) at \(\optState[0]\), then the transversality condition \ref{transversality} of Theorem \ref{th:nonsmooth pmp} strengthens to:
	\begin{align*}
		& \derivative{\hamiltonian (\adjoint[0], 0, \optState[0], \optCtrl[0])}{\dummyst[]} = \multiplierstNcs[0]\quad \text{for some }  \multiplierstNcs[0] \in \Ncs_{\admst[0]}(\optState[0]),\\
		& \adjoint[\hor-1] = -\multipliercost \derivative{\cost[\hor]}{\dummyst[]}(\optState[\hor]) -\biggl(\derivative{\sfc[\hor]}{\dummyst[]} (\optState[])\biggr)\transpose \multipliersfc - \multiplierstNcs[\hor] \quad \text{ for some } \multiplierstNcs[\hor] \in  \Ncs_{\admst[\hor]}(\optState[\hor]).
	\end{align*}
	\item \label{eq:smooth Hamiltonian Maximization} 
	If the dynamics \( \dynamics[j]\) in \ref{eq:sys:dynamics} for every \( j \in \set{0, \ldots, \hor-1}\) are smooth with respect to the control variable \( \dummyctrl[] \) at \(\optCtrl[j] \), then the Hamiltonian maximization \ref{Hamiltonian Maximization} condition in Theorem \ref{th:nonsmooth pmp}  at \(t=j\) simplifies to:
	\begin{align*}
		\inprod{\derivative{\hamiltonian (\adjoint[j], j, \optState[j], \optCtrl[j])}{\dummyctrl[]}}{\dirctrl} \leq 0 
		\quad \text{ for all } \dirctrl \in \Tcs_{\admctrl[j]}(\optCtrl[j]).
	\end{align*}

\end{enumerate}
\end{corollary}

\begin{proof}
	\begin{enumerate}[label={\textup{(\roman*)}}, leftmargin=*, align=left, widest=iii]
	\item For \( i \in \set{1, \ldots, \hor-1} \) if the dynamics \ref{eq:sys:dynamics} is smooth with respect to \(\dummyst[]\) at \(\optState[i]\) then Hamiltonian \( \hamiltonian\) in \eqref{e:hamiltonian} is smooth with respect to \(\dummyst[]\) at \(\optState[i]\).  Since the Hamiltonian is smooth at \( \optState[i]\), the directional derivative of \(\hamiltonian (\adjoint[i], i, \cdot, \ctrl[i])\) at \( \optState[i] \) in any direction \(\dirst \in \R^\dimst\) is the inner product of the gradient of 
	\( \hamiltonian (\adjoint[i], i, \cdot, \ctrl[i]) \) at \( \optState[i] \) and \( \dirst\). In other words, referring to the definition of the Hamiltonian \(\hamiltonian\) in \eqref{e:hamiltonian} and the directional derivative in \eqref{eq:directional derivative},
	\begin{equation*}
	\begin{aligned}
 		\dirder[\dirst]{\hamiltonian (\adjoint[i], i, \cdot,\optCtrl[i])}(\optState[i])=\inprod{\derivative{\hamiltonian (\adjoint[i], i, \optState[i], \optCtrl[i])}{\dummyst[]}}{\dirst}.
 	\end{aligned}
 	\end{equation*}	
 	The condition \ref{state and adjoint_equation} of Theorem \ref{th:nonsmooth pmp} on adjoint dynamics gives us
 		\[
 			\inprod{\adjoint[i-1]	-\derivative{\hamiltonian (\adjoint[i], i, \optState[i], \optCtrl[i])}{\dummyst[]}+ \multiplierstNcs[i]  }{\dirst}\geq  0 \quad	\text{ for all } \dirst \in \R^\dimst.			   
 		\]
 		Since the inner product is linear in \( \dirst \) and the preceding inequality is true for all \( \dirst \in \R^\dimst \), we may replace \(\dirst\) with \(-\dirst \) in the above inequality, leading to
 		\[
 			\adjoint[i-1]-\derivative{\hamiltonian (\adjoint[i], i, \optState[i], \optCtrl[i])}{\dummyst[]}+ \multiplierstNcs[i]=0
 		\]
		as desired.
 	\item For \(t=0\) if  the dynamics \ref{eq:sys:dynamics} is smooth with respect to state variable \( \dummyst[]\) at \(\optState[0]\) then the Hamiltonian \(\hamiltonian\) in \eqref{e:hamiltonian} is continuously differentiable with respect to \( \dummyst[]\) at \(\optState[0]\). Therefore, the directional derivative of Hamiltonian \(\hamiltonian (\adjoint[0],0,\cdot,\ctrl[0])\) at \( \optState[0] \) in any direction \(\dirst \in \R^\dimst\) is
	\begin{equation*}
	\begin{aligned}	
	\dirder[\dirst]{\hamiltonian (\adjoint[0], 0, \cdot, \optCtrl[0])}(\optState[0]) & =               \inprod{\derivative{\hamiltonian (\adjoint[0], 0, \optState[0], \optCtrl[0])}{\dummyst[]}}{\dirst}.
	\end{aligned}
	\end{equation*}
	Then the transversality condition \ref{transversality} of Theorem \ref{th:nonsmooth pmp} and the linearity of the above inner product with respect to \( \dirst \) leads to
	\[
			\derivative{\hamiltonian (\adjoint[0], 0, \optState[0], \optCtrl[0])}{\dummyst[]} = \multiplierstNcs[0] .
	\]
	\item For \( j \in \set{ 0, \ldots, \hor-1}\) if the dynamics \( \dynamics[j]\) in \ref{eq:sys:dynamics} is smooth with respect to the control variable \( \dummyctrl[]\) at \( \optCtrl[j]\), then we can write the directional derivative of the Hamiltonian \( \hamiltonian (\adjoint[j], j, \optState[j], \cdot) \) with respect to control variable \( \dummyctrl[]\) at \( \optCtrl[j]\) along the direction \( \dirctrl \) as the inner product
	\[
		\dirder[\dirctrl]{\hamiltonian (\adjoint[j], j, \optState[j], \cdot)}(\optCtrl[j]) = \inprod{\derivative{\hamiltonian (\adjoint[j], j, \optState[j], \optCtrl[j])}{\dummyctrl[]}}{\dirctrl}.
	\]
	Then for \( t= j\) the Hamiltonian maximization condition \ref{Hamiltonian Maximization} of Theorem \ref{th:nonsmooth pmp} specializes to
	\begin{equation*}
	\inprod{\derivative{\hamiltonian (\adjoint[j], j, \optState[j], \optCtrl[j])}{\dummyctrl[]}}{\dirctrl} \leq 0. \qedhere
	\end{equation*}
	\end{enumerate}
\end{proof}

\begin{remark}

\label{cor:smooth sys PMP}
		Suppose that in \eqref{e:DTOCPNSD} we replace the Lipschitz continuous dynamics hypothesis \ref{eq:sys:dynamics} of \eqref{eq:sys_dynamics} with each
		\[
		\R^\dimst \times \R^\dimctrl \ni (\dummyst[], \dummyctrl[]) \mapsto \dynamics (\dummyst[], \dummyctrl[]) \in \R ^\dimst \quad \text{ for } t = 0, \ldots, \hor-1,
		\]
		being continuously differentiable, then the Hamiltonian as defined in \eqref{e:hamiltonian} is  continuously differentiable in \( \dummyst[]\) and \(\dummyctrl[] \). In this setting, the conditions \ref{state and adjoint_equation}, \ref{transversality} and \ref{Hamiltonian Maximization} of Theorem 
		\ref{th:nonsmooth pmp} can be strengthened to \ref{smooth state and adjoint_equation_2}, \ref{smooth transversality_2} and \ref{smooth Hamiltonian Maximization} respectively, given below:
		\begin{enumerate}[label={\textup{(\roman*\(^{\ast}\))}}, leftmargin=*, align=left, widest=iii]
			\setcounter{enumi}{2}
			\item \label{smooth state and adjoint_equation_2} state and adjoint dynamics:
			\begin{align*}
				& \optState[t+1] =   \derivative{\hamiltonian (\adjoint, t, \optState, \optCtrl)}{\dummyadjoint} \quad \text{ for }t = 0, \ldots, \hor-1,
				\\ 
				& \adjoint[t-1]  = \derivative{\hamiltonian (\adjoint, t, \optState, \optCtrl)}{\dummyst[]}- \multiplierstNcs \\  
				&\quad \quad \quad \text{for some }  \multiplierstNcs \in \Ncs_{\admst}(\optState) \text{ and for }t=1 \ldots, \hor-1;
			\end{align*}		
			\item \label{smooth transversality_2} transversality conditions:
			\begin{align*}
				& \derivative{\hamiltonian (\adjoint[0], 0, \optState[0], \optCtrl[0])}{\dummyst[]} = \multiplierstNcs[0]\quad \text{for all } \dirst \in \R^\dimst \text{ and } \text{for some }  \multiplierstNcs[0] \in \Ncs_{\admst[0]}(\optState[0]),\\
				& \adjoint[\hor-1] = -\multipliercost \derivative{\cost[\hor]}{\dummyst[]}(\optState[\hor]) -\biggl(\derivative{\sfc[\hor]}{\dummyst[]} (\optState[])\biggr)\transpose \multipliersfc - \multiplierstNcs[\hor] \quad \text{ for some } \multiplierstNcs[\hor] \in  \Ncs_{\admst[\hor]}(\optState[\hor]);
			\end{align*}
			
			\item \label{smooth Hamiltonian Maximization} Hamiltonian maximization condition:
			\begin{align*}
				\inprod{\derivative{\hamiltonian (\adjoint, t, \optState, \optCtrl)}{\dummyctrl[]}}{\dirctrl} \leq 0 
				\quad \text{ for all } \dirctrl \in \Tcs_{\admctrl}(\optCtrl)  \text{ and for }t=0 \ldots, \hor-1.
			\end{align*}
		\end{enumerate}

\end{remark}
We have  a second (immediate) special case, whose proof follows at once from the preceding discussion.

\begin{corollary}
\label{remark:linear sys}
	If the controlled system in \eqref{e:DTOCPNSD} is linear, that is, \eqref{eq:sys_dynamics} is replaced by
	\[
		\linearsystem \quad \text{ where } A_t\in \R^{\dimst\times\dimst} \text{ and } \quad B_t\in \R^{\dimst\times\dimctrl}  \quad \text{for } t=0,\ldots,\hor-1,
	\]
	then for the Hamiltonian defined in \eqref{e:hamiltonian}, the assertions of Theorem \ref{th:nonsmooth pmp} hold with the adjoint dynamics in \ref{state and adjoint_equation} given by
	\begin{equation}
	\label{eq:adjoint linearsys}
		\adjoint[t-1]=A_{t} \transpose \adjoint - \multipliercost \derivative{\cost}{\dummyst}(\optState,\optCtrl) - \sfc(\optState) \transpose \multipliersfc- \multiplierstNcs \quad \text{ for } t=0,\ldots,\hor-1.
	\end{equation}
	Moreover, if each \(\admctrl\) is non-empty, convex, and compact, then the condition \ref{Hamiltonian Maximization} in Theorem \ref{th:nonsmooth pmp} becomes the standard Hamiltonian maximization condition
	\[
		\hamiltonian (\adjoint, t, \optState, \optCtrl) = \max_{\dummyctrl[] \in \admctrl} \hamiltonian (\adjoint, t, \optState, \dummyctrl[]) \quad \text{ for } t=0,\ldots,\hor-1.		
	\]
\end{corollary}

\revision{
	The following result addresses the optimal control problem \eqref{e:DTOCPNSD} under a different set of hypotheses than Assumption \ref{assum:regularity of dynamics}. Here we assume that 
	\begin{itemize}
		\item the dynamics are smooth and 
		\item the cost functions are regular but may be nonsmooth
	\end{itemize}
	while retaining the rest of the problem data. We observe that the necessary conditions for this modified problem are similar to the necessary conditions in Theorem \ref{th:nonsmooth pmp} except the transversality conditions; these necessary condition cater to, e.g., \(\ell_1\)-minimization problems that may be employed to enforce sparsity. The precise statement is as follows:

	\begin{theorem}
		\label{th2:ns-Cost_S-Sys}
		Consider the problem \eqref{e:DTOCPNSD} with the following modifications:
		\begin{enumerate}[label={\textup{(\alph*)}}, leftmargin=*, align=left]
			\item \label{th2:smooth:sysDyn} the functions \( \dynamics \) in \eqref{eq:sys_dynamics}  are continuously differentiable everywhere, and  
			\item \label{th2:nonsmooth:cost} the functions  \( \cost\) in \ref{eq:sys:cost} are regular at every \( (\dummyst[], \dummyctrl[]) \in  \R^{\dimst \times \dimctrl}\).\footnote{We emphasize that the assumption of continuous differentiability of the cost functions is being removed; consequently, \( \cost \) may fail to be differentiable at some \((\dummyst[], \dummyctrl[]) \in \R^{\dimst \times \dimctrl}\).}
		\end{enumerate}
		Let \(\seq{\optCtrl}{t}{0}{\hor-1}\) be a control trajectory that solves the optimal optimal control  problem \eqref{e:DTOCPNSD} with  the modifications \ref{th2:smooth:sysDyn} and \ref{th2:nonsmooth:cost}, and let \(\optState[] \Let \seq{\optState}{t}{0}{\hor}\) be the corresponding state trajectory. Then, for the Hamiltonian defined in \eqref{e:hamiltonian}, there exist
		\begin{itemize}
			\item an adjoint trajectory \(\seq{\adjoint}{t}{0}{\hor-1} \subset \dualspace{\R^\dimst}\),
			\item a sequence \(\seq{\multiplierstNcs}{t}{0}{\hor} \subset \dualspace{\R^\dimst}\), and
			%\seq{\multiplierctrlNcs}{t}{0}{\hor-1} \subset \dualspace{\R^\dimctrl}\), %\(\seq{\multiplierctrlNcs}{t}{0}{\hor} \subset \dualspace{\R^\dimctrl}\),
			\item a triplet \(\bigl(\multipliercost, \multipliersfc, \multipliercfc \bigr) \in \R \times \dualspace{\R^{\dimstconstraints}} \times \dualspace{\R^{\dimctrlconstraints}}\),
		\end{itemize} satisfying
		\begin{enumerate}[label={\textup{(\roman*)}}, leftmargin=*, align=left, widest=iii]
			\item \label{th2:non-negativity} the non-negativity condition \(\multipliercost \in \set{0, 1};\)
			\item \label{th2:non-triviality} the non-triviality condition the adjoint trajectory \(\seq{\adjoint}{t}{0}{\hor-1}\) and the triplet \(\bigl(\multipliercost, \multipliersfc, \multipliercfc \bigr)\) do not simultaneously vanish;
			\item \label{th2:state and adjoint_equation} the state and adjoint dynamics \ref{state and adjoint_equation};
%			\begin{align*}
%				& \optState[t+1]  =   \derivative{\hamiltonian (\adjoint, t, \optState, \optCtrl)}{\dummyadjoint} \quad \text{for } t = 0, \ldots, \hor-1, \\
%				& \inprod{\adjoint[t-1]}{\dirst}  \geq \dirder[\dirst]{\hamiltonian (\adjoint,t,\cdot, \optCtrl)}(\optState)-\inprod{ \multiplierstNcs  }{\dirst}\\
%				& \qquad\qquad\qquad\text{for all \(y \in \R^{\dimst}\) and some \( \multiplierstNcs \in \Ncs_{\admst}(\optState)\),  for \( t=1, \ldots, \hor-1\)};
%			\end{align*}	  
			\item \label{th2:transversality} the transversality conditions
				\begin{align*}
					& \dirder[\dirst]{\hamiltonian (\adjoint[0], 0, \cdot, \optCtrl[0])}(\optState[0]) - \inprod{\multiplierstNcs[0]}{\dirst} \leq 0 \nonumber \\
					& \qquad\qquad\qquad\text{ for all } y \in \R^{\dimst} \text{ and some }\multiplierstNcs[0] \in \Ncs_{\admst[0]}(\optState[0]), \text{ and } \nonumber \\
					& \inprod{\adjoint[\hor-1]}{\dirst}  + \multipliercost \dirder[\dirst]{\cost[\hor]}(\optState[\hor]) + \inprod{\biggl(\derivative{\sfc[\hor]}{\dummyst[]} (\optState[])\biggr) \transpose \multipliersfc + \multiplierstNcs[\hor]}{\dirst} \geq 0 \nonumber \\
					& \qquad\qquad\qquad\text{ for all } y \in \R^{\dimst} \text{ and some } \multiplierstNcs[\hor] \in \Ncs_{\admst[\hor]}(\optState[\hor]) \nonumber;
				\end{align*}
			\item \label{th2:Hamiltonian Maximization} the Hamiltonian maximization condition \ref{Hamiltonian Maximization};
%			\begin{align*}
%				\dirder[\dirctrl]{\hamiltonian (\adjoint, t, \optState, \cdot)}(\optCtrl) \leq 0\quad
%				\text{ for all }  \dirctrl \in \Tcs_{\admctrl}(\optCtrl) \;\text{ and for } t=0 \ldots, \hor-1;
%			\end{align*}	  
			\item \label{th2:State Frequency Constraints} frequency constraints on the state trajectory \seq{\optState}{t}{0}{\hor} \ref{State Frequency Constraints};
%			\[
%			\stconstraints(\optState[0], \ldots, \optState[\hor]) = 0;
%			\]
			\item \label{th2:Control Frequency Constraints} frequency constraints on the control action trajectory \seq{\optCtrl}{t}{0}{\hor-1} \ref{Control Frequency Constraints}.
%			\[
%			\ctrlconstraints(\optCtrl[0], \ldots, \optCtrl[\hor-1]) = 0.
%			\]
		\end{enumerate}
		%, the assertions in Theorem \ref{th:nonsmooth pmp} hold with the transversality condition in \ref{transversality}  given by:
		%	exist \begin{itemize}
		%		\item an adjoint trajectory \(\seq{\adjoint}{t}{0}{\hor-1} \subset \dualspace{\R^\dimst}\),
		%		\item a sequence \(\seq{\multiplierstNcs}{t}{0}{\hor} \subset \dualspace{\R^\dimst}\), and
		%		%\seq{\multiplierctrlNcs}{t}{0}{\hor-1} \subset \dualspace{\R^\dimctrl}\), %\(\seq{\multiplierctrlNcs}{t}{0}{\hor} \subset \dualspace{\R^\dimctrl}\),
		%		\item a triplet \(\bigl(\multipliercost, \multipliersfc, \multipliercfc \bigr) \in \R \times \dualspace{\R^{\dimstconstraints}} \times \dualspace{\R^{\dimctrlconstraints}}\),  
		%	\end{itemize} satisfying \ref{non-negativity}, \ref{non-triviality}, \ref{state and adjoint_equation},\ref{Hamiltonian Maximization}, \ref{State Frequency Constraints}, \ref{Control Frequency Constraints}, and  \\
		%	%\begin{enumerate}[label={\textup{(\roman*)}}, leftmargin=*, align=left, widest=iii]
	\end{theorem}
	\begin{proof}
		A proof of the above theorem follows the steps of the proof of the Theorem \ref{th:nonsmooth pmp} given in the next section; we omit the details in the interest of brevity. The proof will follow the steps similar to the proof of  Theorem \ref{th:nonsmooth pmp} present in the next section. To summarize, we start with Step I in \secref{subsec:Step1}, and follow Step II in \secref{subsec:Step2} until the Claim. In the Claim itself, the nonsmooth cost and smooth dynamics are changed from \eqref{eq:optimizationCondition} into the following condition:
		\[
			0 \in \multipliercost \ggrad{\Cost}{\optProc}+  \derivative{\inprod{\Multiplierfd}{\fdy(\cdot)}}{\proc}(\optProc) + \biggl(\derivative{\Sfc}{\proc}(\optProc)\biggr) \transpose \multipliersfc + \biggl(\derivative{\Cfc}{\proc}(\optProc)\biggr) \transpose \multipliercfc + \Ncs_{\admProc}(\optProc).
		\]
		For the above new condition, we define the function \(h(z) \Let \sum_{t=0}^{\hor-1}h_{t}(z) = \sum_{t=0}^{\hor-1}\multipliercost \cost(\st, \ctrl) \) and continue Step II, followed by Step III in \secref{subsec:Step3} to arrive at the end result.
	\end{proof}
}

	%\section{Flow Chart}
	%\label{flowchart:proofsketch}
	%	\input{flowchartproof}
	\section{Proof of Theorem \lowercase{\ref{th:nonsmooth pmp}}}
	\label{sec:proof}
		In this section we provide a detailed proof of the main result Theorem \ref{th:nonsmooth pmp}. Flowchart \ref{flowchart:proofsketch} gives an idea of a proof, and we elaborate the main steps below:

\textbf{Sketch of the proof:} We proceed as per the three steps below:
\begin{enumerate}[label={\textup{Step (\Roman*)}}, leftmargin=*, align=left, widest=iii]
 \item \label{step1:Constructing optimization problem}  
       Our optimal control problem is lifted to an equivalent optimization problem in a suitable high-dimensional product space.
\item \label{step2:First order necessary condition} 
       First order necessary conditions for the optimization problem in \ref{step1:Constructing optimization problem} are obtained using Clarke's necessary condition for non-smooth optimization problems.
\item  \label{step3: projection to original problem}
       The necessary conditions obtained in \ref{step2:First order necessary condition} are projected to appropriate factors to arrive at necessary condition of the original control problem. 
\end{enumerate}
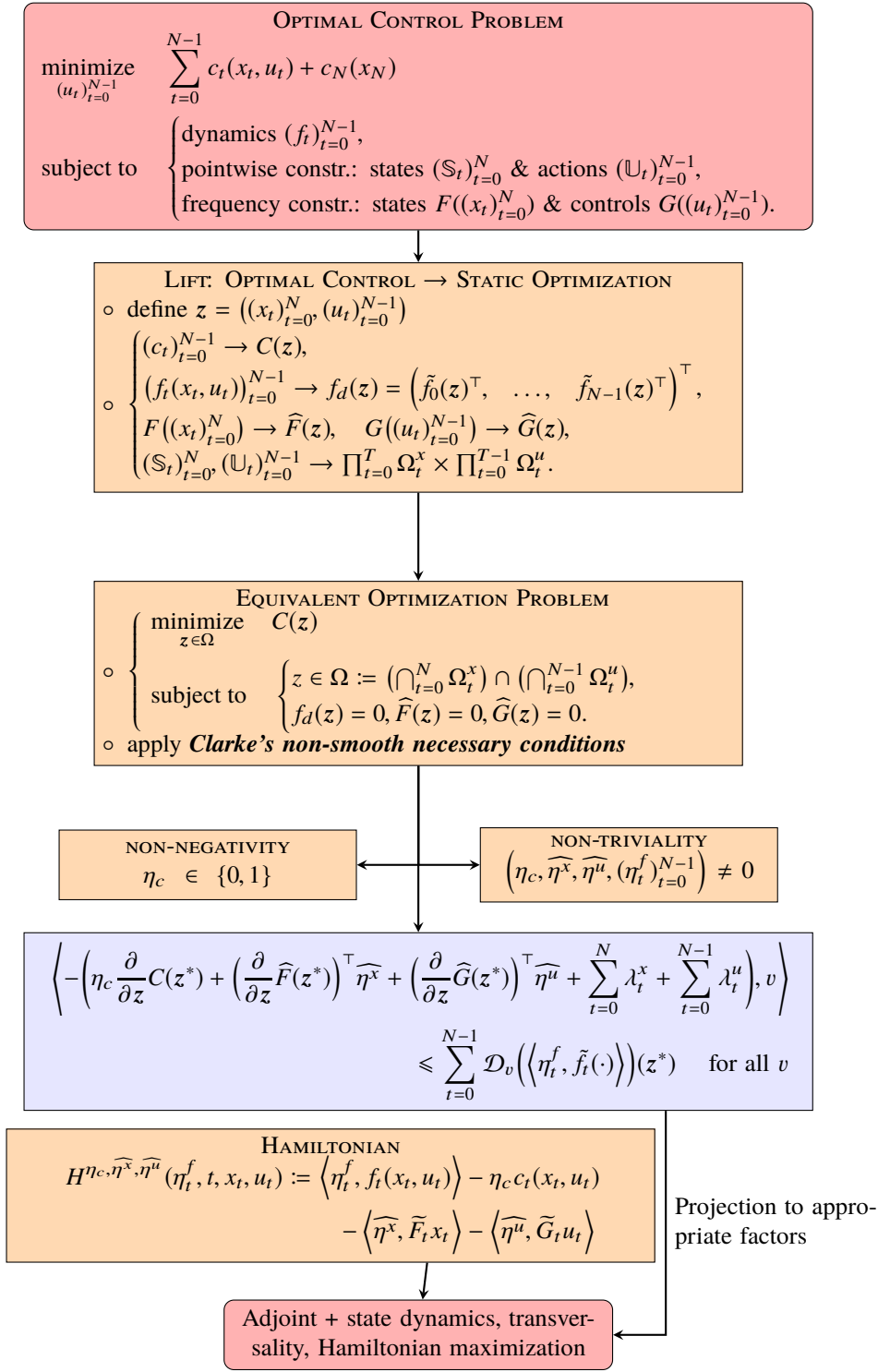
\begin{figure}
		\begin{tikzpicture}[node distance=2cm, auto]
	\node (optimal control problem) [startstop] {$  \textsc{ Optimal Control Problem } $ \\
								$\begin{aligned}                                
                           		& \minimize_{\seq{\ctrl}{t}{0}{\hor-1}} && \sum_{t=0}^{\hor-1} \cost(\st, \ctrl)+\cost[\hor](\st[\hor])\\
                           		& \sbjto &&
                           		\begin{cases}
                           			\text{dynamics } \seq{\dynamics}{t}{0}{\hor-1},\\
                           			\text{pointwise constr.: states \seq{\admst}{t}{0}{\hor} \& actions \seq{\admctrl}{t}{0}{\hor-1}},\\
                           			\text{frequency constr.: states \stconstraints(\seq{\st}{t}{0}{\hor}) \& controls \ctrlconstraints(\seq{\ctrl}{t}{0}{\hor-1})}.
                           		\end{cases}
                           	\end{aligned}
                          $};
	\node (lift) [process, below of=optimal control problem, yshift=-1.75cm] 
		{\textsc{Lift: Optimal Control \(\to\) Static Optimization}\\
			\begin{itemize}[label=\(\circ\), leftmargin=*]
				\item define $ \proc = \bigl( \seq{\st}{t}{0}{\hor},\seq{\ctrl}{t}{0}{\hor-1} \bigr) $\\
                \item $\begin{cases}
            	\seq{\cost}{t}{0}{\hor-1} \to \Cost(\proc), \\
            	\bigl(\dynamics(\st,\ctrl)\bigr)_{t=0}^{\hor-1} \to \fdy(\proc)=\pmat{\dylifted[0](\proc) \transpose, & \ldots, & \dylifted[\hor-1](\proc) \transpose} \transpose,\\
            	\stconstraints\bigl(\seq{\st}{t}{0}{\hor}\bigr) \to \Sfc(\proc),\quad\ctrlconstraints\bigl(\seq{\ctrl}{t}{0}{\hor-1}\bigr) \to \Cfc (\proc),\\
				(\admst)_{t=0}^{\hor}, (\admctrl)_{t=0}^{\hor-1} \to \prod_{t=0}^{T}\admstHD_{t}\times\prod_{t=0}^{T-1}\admctrlHD_{t}.
				\end{cases}$ %
				%\item projection maps
				%\(\proc \mapsto \stprojection(\proc)=\st, \quad \proc \mapsto \ctrlprojection(\proc)=\ctrl\).
            	
            	\
        \end{itemize}
        };
     
    %\node(proc1.b)[substep, right of=lift, xshift=4cm, yshift=-1.5cm]
    %	{\textbf{Projections}
    %	 $ \begin{aligned} 
    %	 \\
    %	  \stprojection(\proc)=\st
    %	 \end{aligned}$
    %	};
    \node(optimization problem)[process, below of=lift, yshift=-2.25cm]
    	{\textsc{ Equivalent Optimization Problem\\}
    	\begin{itemize}[label=\(\circ\), leftmargin=*]
			\item $\left\{\begin{aligned}
    			&\minimize_{\proc \in \admProc}  && \Cost(\proc) \\
    			&\sbjto  &&\begin{cases}
    				z\in \admProc  \Let \bigl(\bigcap_{t=0}^{\hor} \admstHD_{t} \bigr) \cap \bigl(\bigcap_{t=0}^{\hor-1} \admctrlHD_{t} \bigr),\\
    				\fdy (\proc) = 0, \Sfc (\proc) = 0, \Cfc (\proc) = 0.\\
    		\end{cases}
    	\end{aligned}\right.$\\
    	%\item this optimization problem \(\Leftrightarrow\) the original optimal control problem\\
		\item apply \embf{Clarke's non-smooth necessary conditions}  \\
    	 \end{itemize}
    	};
\node(con1)[substep, below of=optimization problem, xshift=-3cm, yshift=-0.75cm]{\textsc{non-negativity}\\$\multipliercost \in \{0,1 \}$};
\node(con2)[substep, below of=optimization problem, xshift=3cm, yshift=-0.75cm]{\textsc{non-triviality}\\$\Bigl(\multipliercost,\multipliersfc,\multipliercfc,\seq{\adjoint}{t}{0}{\hor-1}\Bigr)\neq0$};
\node(con3)[bigprocess, below of= optimization problem, yshift=-3cm]{ 
    		 $\begin{aligned}
    			\inprod{-\biggl(\multipliercost \derivative{\Cost}{\proc}(\optProc) + \Bigl(\derivative{\Sfc}{\proc}(\optProc)\Bigr) \transpose \multipliersfc +  \Bigl(\derivative{\Cfc}{\proc}(\optProc)\Bigr) \transpose \multipliercfc + \sum_{t=0}^{\hor}\MultiplierSt + \sum_{t=0}^{\hor-1}\MultiplierCon \biggr)}{\dir} \\ \leq  \sum_{t=0}^{\hor-1}\dirder{\Bigl(\inprod{\adjoint}{\dylifted(\cdot)}\Bigr)} (\optProc)%{\funInprod_t}(\optProc) %\inprod{\Gderivative{\funInprod}(\optProc)}{\dir} 
    			\quad \text{ for all } \dir 
    			\end{aligned}$
    	     	};
\node(Hamiltonian) [process,  left of=con3, yshift=-2.5cm, xshift=0.75cm]
    	{\textsc{Hamiltonian}
    	$\begin{aligned}
    	 \hamiltonian (\adjoint, t, \st, \ctrl) \Let \inprod{\adjoint}{\dynamics (\st, \ctrl)} &-\multipliercost \cost (\st, \ctrl)\\ -\inprod{\multipliersfc}{\sfc \st} &-\inprod{\multipliercfc}{\cfc \ctrl}
    	 \end{aligned}$
    	};

\node(end)[stop, below of=con3, yshift=-2.5cm]
	{Adjoint + state dynamics, transversality, Hamiltonian maximization};
    %\node(stprojection)[process, below of=necessary conditions,  yshift=-2cm]
    %	{
    %	$\begin{aligned}
    %	\inprod{-\multipliercost \derivative{\cost}{\dummyst[]}(\optState,\optCtrl) - \bigl(\derivative{\sfc}{\dummyst[]}(\optState[])\bigr) \transpose \multipliersfc -  \multiplierstNcs  }{\dir^{\st[]}_{t}} & \leq 
    %	\inprod{\adjoint[t-1]}{\dir^{\st[]}_{t}} - \inprod{\adjoint}{\dirder[\dir^{\st[]}_{t}]{\dynamics(\cdot,\optCtrl)}(\optState)} \\ 
    %	 \text{ for all } \dir^{\st[]}_{t} \in \R^\dimst,\\
    %	\inprod{ -\derivative{\cost}{\dummyctrl[]} \bigl( \optState, \optCtrl \bigr)-\derivative{\cfc}{\dummyctrl[]}(\optState,\optCtrl)- \multiplierctrlNcs}{\dir^{\ctrl[]}_t} & \leq 
    %	 - \inprod{\adjoint}{\dirder[\dir^{\ctrl[]}_t]{\dynamics(\optState,\cdot)}(\optCtrl)} \\ 
    %	     	\text{ for all } \dir^{\ctrl[]}_t\in \R^\dimctrl.
    %	         \end{aligned}$  	
    %	};
\draw[arrow](optimal control problem) --(lift);
\draw[arrow](lift)--(optimization problem);
\draw[arrow](optimization problem)|-(con1);
\draw[arrow](optimization problem)|-(con2);
\draw[arrow](optimization problem)--(con3);
\draw[arrow](Hamiltonian.325)--(end);
\draw[arrow](con3.340)|-node [label, near start] {Projection to appropriate factors} (end);
%\draw[arrow](necessary conditions)--(stprojection);
%\draw[arrow](necessary conditions)|-(ctrlprojection);
\end{tikzpicture}
\caption{Flowchart of a proof}
\label{flowchart:proofsketch}
\end{figure}
	
\subsection{\ref{step1:Constructing optimization problem}} \textbf{Equivalent optimization problem:}
\label{subsec:Step1} 
The objective of this step is to transform the original optimal control problem \eqref{e:DTOCPNSD} to an equivalent optimization problem. The approach is to lift the optimal control problem to an appropriate high-dimensional product space. Here by "lift" we mean the concatenation of \( \hor+1 \) vectors (say \( \dummyst[0], \ldots, \dummyst[\hor]\)) from the space \(\R^\dimst\) corresponding to the states  and  \( \hor \) vectors (say \( \dummyctrl[0], \ldots, \dummyctrl[\hor-1]\)) from the space \(\R^\dimctrl\) corresponding to the control actions. Thus, every vector in the lifted high-dimensional space \(\R^\dimProc\) is of the form 
 \begin{equation}\label{eq:lifted vector}
 \proc \Let (\dummyst[0], \ldots, \dummyst[\hor], \dummyctrl[0], \ldots, \dummyctrl[\hor-1]) \in \R^{\dimProc},
 \end{equation} where \(\dimProc \Let \dimst(\hor +1)+\dimctrl\hor, \ \seq{\dummyst}{t}{0}{\hor}\subset \R^\dimst\) and \(\seq{\dummyctrl}{t}{0}{\hor-1}\subset \R^\dimctrl\).
 For brevity of notation, we define \(\proc \Let \bigl(\bar{\dummyst[]}, \bar{\dummyctrl[]}),\) where
\(
	\bar{\dummyst[]} \Let (\dummyst[0], \ldots, \dummyst[\hor]) \in \R^{\dimst (\hor+1)} \text{ and } \bar{\dummyctrl[]} \Let (\dummyctrl[0], \ldots, \dummyctrl[\hor-1]) \in \R^{\dimctrl \hor}.
\)

In order to extract a vector \(\dummyst \in \R^\dimst \) for some \( t \in \{ 0,\ldots,\hor \} \) and  \(\dummyctrl \in \R^\dimctrl\) for some \( t \in \{0,\ldots,\hor-1 \}\) from \(\proc\), we employ standard state projection maps and control projection maps defined by
\begin{equation}\label{eq:projection}
\begin{aligned}
	\R^{\dimProc} \ni \proc = (\bar{\dummyst[]}, \bar{\dummyctrl[]}) & \mapsto \stprojection (\proc) \Let \dummyst \in \R^{\dimst} \quad \text{for } t = 0, \ldots, \hor,\\
	\R^{\dimProc} \ni \proc =(\bar{\dummyst[]}, \bar{\dummyctrl[]}) & \mapsto \ctrlprojection (\proc) \Let \dummyctrl \in \R^{\dimctrl} \quad \text{for } t = 0, \dots, \hor-1,
\end{aligned}
\end{equation}
  In the above notation the superscript \(\st[] \) or \(\ctrl[]\) of the projection map \(\pi\) indicates the space of states \(\st[]\) or control action \(\ctrl[]\) respectively, and the subscript denotes the time instance.

Let us define the functions and the sets involved in the optimal control problem in the space \(\R^\dimProc\) to arrive at an equivalent optimization problem.
\begin{itemize}[label=\( \circ\), leftmargin=*, align=left]
\item \emph{The lift of the total cost:} Define the function 
		\begin{equation}\label{eq:CostHD}
			\R^{\dimProc} \ni \proc \mapsto \Cost(\proc) \Let \Cost(\bar{\dummyst[]}, \bar{\dummyctrl[]}) = \sum_{t=0}^{\hor-1} \cost(\dummyst, \dummyctrl)+  \cost[\hor](\dummyst[T]) \in \R.
		\end{equation}
\item \emph{The lift of the dynamics:} Define the function 
		\begin{align}\label{eq:process}
			\R^{\dimProc}\ni \proc \mapsto \fdy(\proc) \Let \pmat{\dylifted[0](\proc) \transpose, & \ldots, & \dylifted[\hor-1](\proc) \transpose} \transpose \in \R^{\dimst \hor},
		\end{align}
		where 
		\begin{align*} %\label{eq:dynamics lifted}
			\R^{\dimProc}\ni \proc \mapsto \dylifted (\proc)  \Let \dummyst[t+1]-\dynamics(\dummyst, \dummyctrl)\in \R^\dimst \text{ for each } t = 0, \ldots \hor-1.
		\end{align*}
		Clearly, if a vector \( \proc=(\dummyst[0], \ldots, \dummyst[\hor], \dummyctrl[0], \ldots, \dummyctrl[\hor-1]) \) belongs to the set given by \(\{ y \in \R^\dimProc \ | \  \fdy(y)=0\}\), with the state and control projections \(\seq{\dummyst}{t}{0}{\hor}\), \(\seq{\dummyctrl}{t}{0}{\hor-1}\) then \(\seq{\dummyst}{t}{1}{\hor}\) is a solution of the dynamical system \eqref{eq:sys_dynamics}, corresponding to the initial condition \(\dummyst[0]\) and the control sequence \(\seq{\dummyctrl}{t}{0}{\hor-1}\). Similarly, if \(\seq{\dummyctrl}{t}{0}{\hor-1}\) is an admissible control sequence and  with \(\seq{\dummyst}{t}{0}{\hor-1}\) being the corresponding solution  of the dynamical system \eqref{eq:sys_dynamics}, for the initial condition \(\dummyst[0]\) then the concatenated vector \(\proc\Let(\dummyst[0],\ldots,\dummyst[\hor],\dummyctrl[0],\ldots,\dummyctrl[\hor-1])\) belongs to the set \(\{ y \in \R^\dimProc \ | \  \fdy(y)=0\}.\)
		Therefore, the dynamics \eqref{eq:sys_dynamics} in the optimal control problem \eqref{e:DTOCPNSD} can be equivalently modeled by an equality constraint \( \fdy(\proc)=0 \) in \(\R^\dimProc\).
\item \emph{The lift of the sets corresponding to point-wise state and control action constraints:} Define the sets 
		\begin{align}
			\label{eq: st constraints HD} & \admstHD_{t} \Let \set[\big]{\proc \in \R^{\dimProc} \suchthat \stprojection (\proc) \in \admst} \quad \text{for } t = 0, \ldots, \hor, \\
			\label{eq: ctrl constraints HD} & \admctrlHD_{t} \Let \set[\big]{\proc \in \R^{\dimProc} \suchthat \stprojection (\proc) \in \admctrl} \quad \text{for } t = 0, \ldots, \hor-1.
		\end{align}
		Observe that if the sets \(\seq{\admst}{t}{0}{\hor}\) and \(\seq{\admctrl}{t}{0}{\hor-1}\) are closed then the corresponding lifted sets \((\admstHD_t)_{t=0}^{\hor}\) and \((\admctrlHD_t)_{t=0}^{\hor-1}\) are also closed. Hence the closedness of sets is preserved under the defined lift.
		
		Further any vector \( z=(\dummyst[0],\ldots,\dummyst[\hor],\dummyctrl[0],\ldots,\dummyctrl[\hor-1]) \in \admstHD_{t} \) for some \(t \in \{0,\ldots,\hor \}\) if and only if the corresponding state projection \( \dummyst \in \R^\dimst \) satisfies the state constraints given by \ref{eq:sys:admstset} (i.e., \(\dummyst \in \admst\)). Similarly for any \(z=(\dummyst[0],\ldots,\dummyst[\hor],\dummyctrl[0],\ldots,\dummyctrl[\hor-1]) \in \admctrlHD_{t} \) for some \(t \in \{0,\ldots,\hor-1 \}\) if and only if the corresponding control projection \( \dummyctrl \in \R^\dimctrl \) satisfies the control constraints given by \ref{eq:sys:admctrlset} (i.e., \(\dummyctrl \in \admctrl\)).
		
		Therefore \(\proc=(\dummyst[0],\ldots,\dummyst[\hor],\dummyctrl[0],\ldots,\dummyctrl[\hor-1])  \in \bigl(\bigcap_{t=0}^{\hor} \admstHD_{t} \bigr) \cap \bigl(\bigcap_{t=0}^{\hor-1}\admctrlHD_{t} \bigr)\) if and only if \( \dummyst \in \admst \) for \( t=0, \ldots, \hor ,\) and \( \dummyctrl \in \admctrl \) for \( t=0, \ldots, \hor-1 \). Hence the point-wise state and control constraints is  equivalently given by  the constraint \(\proc \in \bigl(\bigcap_{t=0}^{\hor} \admstHD_{t} \bigr) \cap \bigl(\bigcap_{t=0}^{\hor-1}\admctrlHD_{t} \bigr)\) in \(\R^\dimProc \).
\item \emph{The lift of the frequency constraints on state and control trajectory:} Define the functions 
		\begin{align}
			\label{eq: st freq constraints HD} & \R^{\dimProc}\ni \proc \mapsto \Sfc (\proc) \Let \stconstraints \bigl(\stprojection[0](\proc), \ldots, \stprojection[\hor](\proc) \bigr)\in \R^{\dimstconstraints} \quad \text{and}\\
			& \label{eq: ctrl freq constraints HD}\R^{\dimProc}\ni \proc \mapsto \Cfc (\proc) \Let \ctrlconstraints \bigl(\ctrlprojection[0] (\proc), \ldots, \ctrlprojection[\hor-1](\proc) \bigr)\in \R^{\dimctrlconstraints}.
		\end{align}
		From the definition of \( \Sfc \) and \( \Cfc \) it is clear that the equality constraints \(\Sfc (\proc)=0\) and \(\Cfc (\proc)=0\) in \(\R^\dimProc \) are equivalent to the state frequency constraints \ref{eq:sys:SFC} and the control frequency constraints  \ref{eq:sys:CFC}, respectively.
\end{itemize} 
In view of the various definitions above, we define the optimization problem 
\begin{equation}\label{e:opt prob}
\boxed{
	\begin{aligned}
		&\minimize_{\proc \in \admProc}  && \Cost(\proc) \\
		&\sbjto  &&\begin{cases}
			\admProc  \Let \bigl(\bigcap_{t=0}^{\hor} \admstHD_{t} \bigr) \cap \bigl(\bigcap_{t=0}^{\hor-1} \admctrlHD_{t} \bigr),\\
			\fdy (\proc) = 0,\\
			\Sfc (\proc) = 0, \\
			\Cfc (\proc) = 0,\\
	\end{cases}
\end{aligned}
}
\end{equation}

Note that the cost and the constraints of \eqref{e:opt prob} and  \eqref{e:DTOCPNSD} are identical; consequently,  these are equivalent problems.

		\subsection{\ref{step2:First order necessary condition}} \textbf{Necessary condition for the equivalent optimization problem:}
\label{subsec:Step2}

 Let \(\optProc\) denote a solution to the optimization problem \eqref{e:opt prob}, comprising of an optimal control sequence \(\seq{\optCtrl}{t}{0}{\hor-1}\) and its corresponding state trajectory \(\seq{\optState}{t}{0}{\hor}\); that is, \( \optProc=(\optState[0],\cdots,\optState[\hor],\optCtrl[0],\cdots,\optCtrl[\hor-1])\). Since \eqref{e:opt prob} is equivalent to the optimal control problem \eqref{e:DTOCPNSD}, if \(\optProc \) is a solution to the optimization problem \eqref{e:opt prob}, then  \(\seq{\optCtrl}{t}{0}{\hor-1}\) is a solution to the optimal control problem \eqref{e:DTOCPNSD}  and \(\seq{\optState}{t}{0}{\hor}\) is its corresponding optimal state trajectory.

 The following theorem (\secref{Clarke's multiplier rule})) provides necessary condition for \(\optProc\) to be a solution of \eqref{e:opt prob}:
 \begin{theorem}\label{Th:Clarke's Opt Condition}
  If \(\optProc\) is a solution to the optimization problem \eqref{e:opt prob}, then there exists a non-trivial vector \(\bigl(\multipliercost, \Multiplierfd, \multipliersfc, \multipliercfc \bigr) \in \{0,1\} \times \dualspace{\R^{\dimst \hor}} \times \dualspace{\R^{\dimstconstraints}} \times \dualspace{\R^{\dimctrlconstraints}}\) such that
\begin{equation}
\label{Opt_problem_necessary_condition}
	0 \in \ggrad{\Big({\multipliercost \Cost(\cdot) +\inprod{\Multiplierfd}{\fdy (\cdot)} + \inprod{\multipliersfc}{\Sfc (\cdot)} + \inprod{\multipliercfc}{\Cfc (\cdot)}}\Big)}{\optProc} +\Ncs_{\admProc}(\optProc).
\end{equation}
\end{theorem} 
In  the setting of Theorem \ref{Th:Clarke's Opt Condition}, the various scalars/vectors, \( \multipliercost,\Multiplierfd, \multipliersfc,  \multipliercfc \) are called multipliers  corresponding to the lifted cost \(\Cost\), the function \( \fdy \), the lifted state frequency constraints \(\Sfc\), and  the lifted control frequency constraints \( \Cfc \), respectively. 

The condition \eqref{Opt_problem_necessary_condition} is a set theoretic necessary condition for optimality, which means that there exist some non-trivial vectors in the sets 
\[
 \ggrad{\Big({\multipliercost \Cost(\cdot) +\inprod{\Multiplierfd}{\fdy (\cdot)} + \inprod{\multipliersfc}{\Sfc (\cdot)} + \inprod{\multipliercfc}{\Cfc (\cdot)}}\Big)}{\optProc} \text{ and } \Ncs_{\admProc}(\optProc) 
 \] such that their sum is zero. In order to simplify the condition above we characterize the elements of these sets. 
%Let us consider that,
%\begin{equation}
%\begin{aligned}
%	\multipliersfc & =(\multipliersfc,\multipliersfc_{\ctrl[]}) \\
%	\multipliercfc & =(\multipliercfc_{\st[]},\multipliercfc)
%\end{aligned}
%\end{equation}

% % % % % % % % % % % % % % % % % % % % % % % %
\begin{itemize}[label=\( \circ\),leftmargin=*, align=left]
\item 
Consider, first, the set \( \Ncs_{\admProc}(\optProc)\). We have the following characterization of \( \Ncs_{\admProc}(\optProc)\):

\embf{Claim:} For any vector \(\gamma \in \Ncs_{\admProc}(\optProc)\), there exist vectors \(\gamma_{t}^{\st[]} \in \Ncs_{\admstHD_{t}}(\optProc)\) and  \(\gamma_{t}^{\ctrl[]} \in \Ncs_{\admctrlHD_{t}}(\optProc)\) such that
\[
	\gamma = \sum_{t=0}^{\hor} \gamma_{t}^{\st[]} + \sum_{t=0}^{\hor-1} \gamma_{t}^{\ctrl[]}.
\]

% \begin{lemma}\label{lemma:adjoint in normal cone is sum of adjoints}
% For the set \( \admProc =(\bigcap_{t=0}^{\hor}\admstHD_{t})\cap (\bigcap_{t=0}^{\hor-1}\admctrlHD_{t})\) 
% defined in \eqref{e:opt prob},
% \[
% 	\Ncs_{\admProc}(\optProc)=	\chull \biggl( \Bigl(\bigcup_{t=0}^{\hor} \Ncs_{\admstHD_{t}}(\optProc) \Bigr) \cup \Bigl( \bigcup_{t=0}^{\hor-1} \Ncs_{\admctrlHD_{t}}(\optProc)\Bigr) \biggr).
% \]
% In particular, any \(\lambda \in \Ncs_{\admProc} (\optProc)\) is expressible as a sum \(\sum_{t=0}^{\hor} \MultiplierSt + \sum_{t=0}^{\hor-1} \MultiplierCon\) for some \(\MultiplierSt \in \Ncs_{\admstHD_{t}} (\optProc)\) and some \(\MultiplierCon \in \Ncs_{\admctrlHD_{t}}(\optProc)\).
% \end{lemma}
\emph{Proof of Claim:}
	We know that the Clarke normal cone to any non-empty set is a closed and convex cone. Therefore, both \(\bigl(\Ncs_{\admstHD_{t}}(\optProc)\bigr)_{t=0}^{\hor}\) and \(\bigl(\Ncs_{\admctrlHD_{t}}(\optProc)\bigr)_{t=0}^{\hor-1}\) are sequences of closed and convex cones. 
	
	Consider the set \(\genSet \Let \chull \biggl( \Bigl(\bigcup_{t=0}^{\hor} \Ncs_{\admstHD_{t}}(\optProc) \Bigr) \cup \Bigl( \bigcup_{t=0}^{\hor-1} \Ncs_{\admctrlHD_{t}}(\optProc)\Bigr) \biggr)\). If \(\genSet\) is not closed, then by Theorem \ref{th:closure of convex hull} there exist vectors \(\MultiplierSt \in \Ncs_{\admstHD_{t}}(\optProc),\) and \(\MultiplierCon \in \Ncs_{\admctrlHD_{t}}(\optProc) \text{ for } t=0,\ldots, \hor-1, \ \MultiplierSt[\hor] \in \Ncs_{\admstHD_{\hor}}(\optProc)\), not all of them zero, such that \(\sum_{t=0}^{\hor}\MultiplierSt+\sum_{t=0}^{\hor-1}\MultiplierCon=0.\)

By the definition of the Clarke normal cone, if \(\MultiplierSt \in \Ncs_{\admstHD_{t}}(\optProc),\) then for any vector \(\dir_t\in \Tcs_{\admstHD_{t}}(\optProc)\) we have \(  \inprod{\MultiplierSt}{\dir_t}\leq 0.\) For the set \(\admstHD_{t}\) defined in \eqref{eq: st constraints HD} 
\(\stprojection[s](\dir_t) \) for \( s \in \set{0, \ldots, \hor} \setminus \set{t}\)  and \(\ctrlprojection[r](\dir_t) \) for  \(r\in \set{0, \ldots, \hor-1} \) are arbitrary, and the corresponding coordinates in \(\MultiplierSt[t]\) are zeros. That is,  \( \stprojection[s](\MultiplierSt)=0 \) for \( s \in \set{0, \ldots, \hor}\setminus \set{t} \) and \(\ctrlprojection[r](\MultiplierSt)=0\) for \(r \in \set{0, \ldots, \hor-1} \). Similarly, it follows that  \( \stprojection[s](\MultiplierCon)=0 \text{ for } s \in \set{0, \ldots, \hor}\) and \(\ctrlprojection[s](\MultiplierCon)=0 \text{ for } s \in \set{0, \ldots, \hor-1} \setminus \set{t}. \) Therefore,  \(\MultiplierSt\) will be of form
\begin{equation}
\begin{aligned}\label{eq:Multipliers in HD}
	\MultiplierSt &= 
			\pmat{0, \ldots, 0, \multiplierstNcs, 0, \ldots, 0} \quad   \text{ for } t = 0, \ldots, \hor,\\
	\MultiplierCon & =\pmat{0, \ldots, 0, \multiplierctrlNcs, 0, \ldots, 0} \quad  \text{ for } t =0, \ldots, \hor-1 ,
		\end{aligned}	
\end{equation}
where \(\multiplierstNcs \Let \stprojection{(\MultiplierSt)}, \  \multiplierctrlNcs \Let \ctrlprojection{(\MultiplierCon)} \text{ for } t = 0, \ldots, \hor-1,  \text{ and }  \multiplierstNcs[\hor]\Let\stprojection[\hor]{(\MultiplierSt[\hor])}\).
Hence, \( \sum_{t=0}^{\hor}\MultiplierSt + \sum_{t=0}^{\hor-1}\MultiplierCon=0 \) implies  \(\MultiplierSt=\MultiplierCon=\MultiplierSt[\hor]=0 \) for \(0 \leq t\leq \hor-1\). This contradicts the non-triviality assertion on \(\seq{\MultiplierSt}{t}{0}{\hor},\seq{\MultiplierCon}{t}{0}{\hor-1}.\)
Therefore, our assumption is not true which means \(\genSet\) is closed.

From Lemma \ref{Tcs to intersection is intersection of Tcs}, \(\Tcs_{\admProc}(\optProc) = \bigl(\bigcap_{t=0}^{\hor} \Tcs_{\admstHD_{t}}(\optProc)\bigr) \cap \bigl(\bigcap_{t=0}^{\hor-1} \Tcs_{\admctrlHD_{t}}(\optProc) \bigr)\). Therefore, from Theorem \ref{th:Dual cone of intersection is union of dual cone} we see that
\[
\Ncs_{\admProc}(\optProc) = \closure \biggl( \chull \biggl( \Bigl(\bigcup_{t=0}^{\hor} \Ncs_{\admstHD_{t}}\Bigr) \cup \Bigl( \bigcup_{t=0}^{\hor-1} \Ncs_{\admctrlHD_{t}}\Bigr) \biggr)\biggr) = \chull \biggl( \Bigl(\bigcup_{t=0}^{\hor} \Ncs_{\admstHD_{t}}\Bigr) \cup \Bigl( \bigcup_{t=0}^{\hor-1} \Ncs_{\admctrlHD_{t}}\Bigr) \biggr).
\]
	The assertion of the claim follows at once from the fact that \( (\Ncs_{\admstHD_{t}})_{t=0}^{\hor}\) and \( (\Ncs_{\admctrlHD_{t}})_{t=0}^{\hor-1} \) are sequences cones. \qed

% % % % % % % % % % % % % % % % % % % % % % % % %

\item Consider, second, the set \( \ggrad{\Big({\multipliercost \Cost(\cdot) +\inprod{\Multiplierfd}{\fdy (\cdot)} + \inprod{\multipliersfc}{\Sfc (\cdot)} + \inprod{\multipliercfc}{\Cfc (\cdot)}}\Big)}{\optProc} \) in \eqref{Opt_problem_necessary_condition}.
The functions \(\Cost, \Sfc ,\Cfc\) in \eqref{e:opt prob} are continuously differentiable. Using Lemma \ref{ggrad_of Lipschitz_+_C1} we get,
\begin{align*}
\label{eq:genGradtermProof}
	& \ggrad\Bigl{{\multipliercost \Cost + \inprod{\Multiplierfd}{\fdy} + \inprod{\multipliersfc}{\Sfc} + \inprod{\multipliercfc}{\Cfc}}\Bigr}({\optProc}) = \nonumber \\
	& \quad \multipliercost \derivative{\Cost}{\proc}(\optProc) + \ggrad{\inprod{\Multiplierfd}{\fdy(\cdot)}}{\optProc} + \Bigl( \derivative{\Sfc}{\proc}(\optProc)\Bigr) \transpose \multipliersfc +\Bigl( \derivative{\Cfc}{\proc}(\optProc) \Bigr) \transpose \multipliercfc.
\end{align*}
In other words, \eqref{Opt_problem_necessary_condition} can be simplified to
\begin{equation}
\label{eq:optimizationCondition}
	0 \in \multipliercost \derivative{\Cost}{\proc}(\optProc)+ \ggrad{\inprod{\Multiplierfd}{\fdy(\cdot)}}{\optProc} + \biggl(\derivative{\Sfc}{\proc}(\optProc)\biggr) \transpose \multipliersfc + \biggl(\derivative{\Cfc}{\proc}(\optProc)\biggr) \transpose \multipliercfc + \Ncs_{\admProc}(\optProc).
\end{equation}
To characterize the elements of the set \(\ggrad{\inprod{\Multiplierfd}{\fdy(\cdot)}}{\optProc}\), let us define the functions
 \begin{equation}  
 \label{eq:function innerprod of dynamics and adjoint}
 \begin{aligned}
 	 \admProc \ni \proc\mapsto \dylifted(\proc) & \Let \dummyst[t+1]-\dynamics(\dummyst,\dummyctrl) \text{ for } t=0,\ldots,\hor-1, \\
	 \admProc \ni \proc\mapsto \funInprod_t(\proc) & \Let \inprod{\stprojection(\Multiplierfd)}{\dylifted(\proc)}=\inprod{\adjoint}{\dylifted(\proc)}\in \R  \text{ for } t=0,\ldots,\hor-1, \\
	 \admProc \ni \proc\mapsto \funInprod(\proc)&\Let\sum_{t=0}^{\hor-1}\funInprod_t(\proc) =\inprod{\Multiplierfd}{\fdy(\proc)}\in \R.   
 \end{aligned}
 \end{equation}
We immediately see that \(\ggrad{\inprod{\Multiplierfd}{\fdy(\cdot)}}{\optProc} = \ggrad{\funInprod}{\optProc}.
\)

If \(\omega \in \ggrad{\funInprod}{\optProc}\), then from the definition of generalized gradient \eqref{def:directional derivative} \(  \inprod{\omega}{\dir}\leq \gdd[\funInprod]{\optProc}{\dir}\)  for all \( \dir \in \R^\dimProc .\) Give that \( \dynamics \)'s are regular functions ( cf. \ref{assum:regularity of dynamics}), the function \( \funInprod \) defined above is also a regular function. Moreover, since \( \funInprod \) is a regular function, \( \funInprod \) has a directional derivative in each direction at \(\optProc\) and \( \gdd[\funInprod]{\optProc}{\dir}=\dirder{\funInprod}(\optProc) =\sum_{t=0}^{\hor}\dirder{\funInprod_t}(\optProc) \). Consequently, 
\begin{equation}\label{ineq:directinal derr lifted dynamics}
\inprod{\omega}{\dir}\leq \sum_{t=0}^{\hor}\dirder{\funInprod_t}(\optProc) \quad  \text{  for all }  \dir \in \R^\dimProc.
\end{equation} 

Let \(\omega \in \ggrad{\inprod{\Multiplierfd}{\fdy(\cdot)}}{\optProc}\) and \(\lambda=\sum_{t=0}^{\hor} \MultiplierSt +\sum_{t=0}^{\hor-1} \MultiplierCon \in \Ncs_{\admProc}(\optProc),\) where \( \MultiplierSt \in \Ncs_{\admstHD_{t}}(\optProc)\), \( \MultiplierCon \in \Ncs_{\admctrlHD_{t}}(\optProc)\), be vectors such that
\[
	\multipliercost \derivative{\Cost}{\proc}(\optProc) + \omega +\biggl(\derivative{\Sfc}{\proc}(\optProc)\biggr) \transpose \multipliersfc + \biggl(\derivative{\Cfc}{\proc}(\optProc)\biggr) \transpose \multipliercfc + \sum_{t=0}^{\hor} \MultiplierSt +\sum_{t=0}^{\hor-1} \MultiplierCon = 0,
 \]
and a transposition leads to
\begin{equation*}
\label{eq:condition for optimization}
	-\Biggl( \multipliercost \derivative{\Cost}{\proc}(\optProc) + \biggl(\derivative{\Sfc}{\proc}(\optProc)\biggr) \transpose \multipliersfc + \biggl(\derivative{\Cfc}{\proc}(\optProc)\biggr) \transpose \multipliercfc + \sum_{t=0}^{\hor}\MultiplierSt + \sum_{t=0}^{\hor-1}\MultiplierCon\Biggr) =  \omega.
\end{equation*}
 Here \(\omega \in \ggrad{\inprod{\Multiplierfd}{\fdy(\cdot)}}{\optProc}\) implies that \(\inprod{\omega}{\dir}\leq \dirder{\funInprod}(\optProc)  \text{ for all } \dir \in \R^\dimProc\), and hence the inequality simplifies \eqref{ineq:directinal derr lifted dynamics} to
\begin{equation}
\label{eq:necessary condition for optimization problem}
\begin{aligned}
	\inprod{-\Biggl(\multipliercost \derivative{\Cost}{\proc}(\optProc) + \biggl(\derivative{\Sfc}{\proc}(\optProc)\biggr) \transpose \multipliersfc +  \biggl(\derivative{\Cfc}{\proc}(\optProc)\biggr) \transpose \multipliercfc + \sum_{t=0}^{\hor}\MultiplierSt + \sum_{t=0}^{\hor-1}\MultiplierCon \Biggr)}{\dir}& \\ \leq  \sum_{t=0}^{\hor}\dirder{\funInprod_t}(\optProc) %\inprod{\Gderivative{\funInprod}(\optProc)}{\dir} 
	\quad \text{ for all } \dir \in \R^\dimProc. &
\end{aligned}
\end{equation}
\end{itemize}
In other words we have established the following proposition:
\begin{proposition}\label{prop: reduced necessary condition}
If \(\optProc\) is a solution of problem \eqref{e:opt prob}, then there exist a non-trivial vector \(\bigl(\multipliercost, \Multiplierfd, \multipliersfc, \multipliercfc \bigr) \in \{0,1\} \times \dualspace{\R^{\dimst \hor}} \times \dualspace{\R^{\dimstconstraints}} \times \dualspace{\R^{\dimctrlconstraints}}\), vector \( \MultiplierSt \in \Ncs_{\admstHD_{t}}(\optProc)\) for each \(t=0,\ldots,\hor\), and vector \( \MultiplierCon \in \Ncs_{\admctrlHD_{t}}(\optProc)\) for each \(t=0,\ldots,\hor-1\) such that \eqref{eq:necessary condition for optimization problem} holds true, where \(\funInprod_t\) is as defined in \eqref{eq:function innerprod of dynamics and adjoint}.
\end{proposition}

The inequality \eqref{eq:necessary condition for optimization problem} obtained is a necessary condition for the solutions of the equivalent optimization problem \eqref{e:opt prob}.
% % % % % % % % % % % % % % % % % % % % % % % % % % % % % % % % % % % % % % % % % % %

%		
		\subsection{\ref{step3: projection to original problem}} \textbf{Projecting the condition in equation \eqref{eq:necessary condition for optimization problem} to the original factor spaces:}
\label{subsec:Step3}

We use Proposition \ref{prop: reduced necessary condition} to arrive at a set of necessary conditions for a solution of the optimal control problem \eqref{e:DTOCPNSD}. It may be observed, in particular, that the non-negativity condition \ref{non-negativity} in the main result follows directly from \eqref{Opt_problem_necessary_condition}.

For some \(t \in \set{0, \ldots, \hor}\), choose a vector \(\dir^{\st[]}_{t}\in \R^\dimst\) and define
\begin{equation} \label{eq:dummy dir state}
	\dummydir[] \Let(0, \ldots, \dir^{\st[]}_{t}, 0, \ldots, 0)\in \R^\dimProc.
\end{equation}
The projections of \(\dummydir[]\) \eqref{eq:projection} are
\begin{equation}\label{eq:projection of dummy direction}
	\stprojection[i]{(\dummydir[])}=\begin{cases}
										\dir_{t}^{\st[]} & \text{for }  i=t,\\
										0  & \text{otherwise,}
									\end{cases} \quad \text{and} \quad  \ctrlprojection[i]{(\dummydir[])} = 0 \quad \text{for } i=0,\ldots,\hor-1.
	\end{equation}
 Substituting \(\dir=\dummydir[]\) in \eqref{eq:necessary condition for optimization problem}, we get
\begin{equation}\label{eq:necessary condition optimization with tildev}
\begin{aligned}
	\inprod{-\biggl(\multipliercost \derivative{\Cost}{\proc}(\optProc) + \Bigl(\derivative{\Sfc}{\proc}(\optProc)\Bigr) \transpose \multipliersfc +  \Bigl(\derivative{\Cfc}{\proc}(\optProc)\Bigr) \transpose \multipliercfc + \sum_{i=0}^{\hor}\MultiplierSt[i] + \sum_{i=0}^{\hor-1}\MultiplierCon[i] \biggr)}{\dummydir[]} \\
	\leq  \sum_{i=0}^{\hor}\dirder[\tilde{\dir}]{\funInprod_i}(\optProc). %\inprod{\Gderivative{\funInprod}(\optProc)}{\dir} 
\end{aligned}
\end{equation}
Let us look at each term individually, starting from the left, in the above inequality.
\begin{itemize}[label=\(\circ\), leftmargin=*]
	\item The first term corresponds to the cost function (\(\Cost\)) of the optimization problem. From the definition \eqref{eq:CostHD} of \(\Cost\), its gradient is 
	\begin{equation*}
	\begin{aligned}
		& \derivative{\Cost}{\proc}(\optProc)= \\
		& \begin{pmatrix}
		\derivative{\cost[0]}{\dummyst[]}(\optState[0],\optCtrl[0])\transpose &
		\cdots  & 
		\derivative{\cost[\hor]}{\dummyst[]}(\optState[\hor])\transpose &
		\derivative{\cost[0]}{\dummyctrl[]}(\optState[0],\optCtrl[0])\transpose &
		\cdots &
		\derivative{\cost[\hor-1]}{\dummyctrl[]}(\optState[\hor-1],\optCtrl[\hor-1])\transpose
		\end{pmatrix}\transpose,\\
	\end{aligned}
	\end{equation*}
	and therefore,
	\[
	\inprod{ \multipliercost \derivative{\Cost}{\proc}(\optProc)}{\dummydir[]}=\multipliercost\left(\sum\limits_{i=0}^{\hor}\inprod{\derivative{\cost[i]}{\dummyst[]}(\optState[i],\optCtrl[i])}{\stprojection[i]{(\dummydir[])}}+\sum\limits_{i=0}^{\hor-1}\inprod{\derivative{\cost[i]}{\dummyctrl[]}(\optState[i],\optCtrl[i])}{\ctrlprojection[i]{(\dummydir[])}}\right).
	\]
	Substituting the projections  of \(\dummydir[]\) from \eqref{eq:projection of dummy direction} gives
	\begin{equation}\label{eq:st projection cost}
	\begin{aligned}
	\inprod{ \multipliercost \derivative{\Cost}{\proc}(\optProc)}{\dummydir[]}=\inprod{ \multipliercost \derivative{\cost[t]}{\dummyst[]}(\optState[t],\optCtrl[t])}{\stprojection[t]{(\dummydir[])}}=\inprod{ \multipliercost \derivative{\cost[t]}{\dummyst[]}(\optState[t],\optCtrl[t])}{\dir^{\st[]}_{t}} .
	\end{aligned}
	\end{equation}
	\item The second term in \eqref{eq:necessary condition optimization with tildev} corresponds to the constraints on the state trajectory (in particular, to the function \(\stconstraints\),) of the original optimal control problem, equivalently represented by the function  \(\Sfc\) in \(\R^{\dimProc}\).
	Recalling the definition \eqref{eq: st freq constraints HD} of \(\Sfc \), we see that 
	\[
	\Sfc(\dummyst[0],\cdots,\dummyst[\hor],\dummyctrl[0],\cdots,\dummyctrl[\hor-1])=\stconstraints(\dummyst[0],\cdots,\dummyst[\hor]) \in \R^{\dimstconstraints}.
	\]
	Let \(\Sfc =\pmat{\Sfc_{1} & \cdots & \Sfc_{\dimstconstraints}}\transpose.\) Its gradient is given by\\
	\[ 
	\derivative{\Sfc}{\proc}
			=\begin{pmatrix}
				\derivative{\Sfc}{\dummyst[0]} & \cdots &  \derivative{\Sfc}{\dummyst[\hor]}& \derivative{\Sfc}{\dummyctrl[0]} & \cdots & \derivative{\Sfc}{\dummyctrl[\hor-1]}
			 \end{pmatrix} \in \R^{\dimstconstraints \times \dimProc}, 
	\]
	where 
	\begin{align*}
			  \derivative{\Sfc}{\dummyst[i]}
			 		 =  & \begin{pmatrix}
			 		\derivative{\Sfc_{1}}{\dummyst[i]^1} & \cdots &  \derivative{\Sfc_{1}}{\dummyst[i]^\dimst}\\  
			 		\vdots & \ddots  &\vdots \\
			 		\derivative{\Sfc_{\dimstconstraints}}{\dummyst[i]^1} & \cdots &  \derivative{\Sfc_{\dimstconstraints}}{\dummyst[i]^\dimst}  
			 		\end{pmatrix}\in \R^{\dimstconstraints\times\dimst} 
			 		\quad \text{and}\\
			 		 \derivative{\Sfc}{\dummyctrl[i]}
			 		  =  & \begin{pmatrix}
			 		\derivative{\Sfc_{1}}{\dummyctrl[i]^1} & \cdots &  \derivative{\Sfc_{1}}{\dummyctrl[i]^\dimctrl}\\  
			 		\vdots & \ddots  &\vdots \\
			 		\derivative{\Sfc_{\dimstconstraints}}{\dummyctrl[i]^1} & \cdots &  \derivative{\Sfc_{\dimstconstraints}}{\dummyctrl[i]^\dimctrl}  
			 		\end{pmatrix}\in \R^{\dimstconstraints\times\dimctrl}.
	\end{align*}
	From the definition of \(\Sfc\) we see at once  that \(\derivative{\Sfc}{\dummyctrl[i]}=0  \) for  \(i= 0, \ldots, \hor-1\) and \(\derivative{\Sfc}{\dummyst[i]}= \derivative{\stconstraints}{\dummyst[i]} \) for \(i= 0, \ldots, \hor\), which shows that
		\[ 
			\derivative{\Sfc}{\proc}
					=\begin{pmatrix}
						\derivative{\stconstraints}{\dummyst[0]} & \cdots &  \derivative{\stconstraints}{\dummyst[\hor]}& \mathbf{0}_{\dimstconstraints \times \dimctrl\hor} 
					 \end{pmatrix} \in \R^{\dimstconstraints \times \dimProc}.
			\]
		Moreover, from the definition \ref{eq:sys:SFC}  of \(\stconstraints,\) we see that \(\derivative{\stconstraints}{\dummyst[i]}=\derivative{\sfc[i]}{\dummyst[]}\) for \(i=0,\ldots,\hor.\) This implies, for a vector \(\multipliersfc \in \R^{\dimstconstraints},\) 
		\[
		\biggl(\derivative{\Sfc}{\proc}\biggr)\transpose\multipliersfc=\begin{pmatrix}
		\Bigl(\derivative{\sfc[0]}{\dummyst[]}\Bigr)\transpose \multipliersfc  \\
		\vdots \\
		\Bigl(\derivative{\sfc[\hor]}{\dummyst[]}\Bigr)\transpose \multipliersfc\\
		\mathbf{0}_{\dimctrl\hor\times 1}
		\end{pmatrix} \in \R^{\dimProc}.
		\]
	Therefore, the second term in \eqref{eq:necessary condition optimization with tildev} for the projection \eqref{eq:projection of dummy direction} of given \(\dummydir[]\)  is 
	\begin{equation}\label{eq:st projection st freq constraints}
	\begin{aligned}
	\inprod{\biggl(\derivative{\Sfc}{\proc}(\optProc)\biggr)\transpose\multipliersfc}{\dummydir[]}
	=\sum\limits_{i=0}^{\hor}\inprod{\biggl(\derivative{\sfc[i]}{\dummyst[]}(\optState[i])\biggr)\transpose \multipliersfc}{\stprojection[i]{(\dummydir[])}} =\inprod{\biggl(\derivative{\sfc[t]}{\dummyst[]}(\optState[t])\biggr)\transpose \multipliersfc}{\dummydir[t]^{\st[]}}.
	\end{aligned}
	\end{equation}
	\item The third term in \eqref{eq:necessary condition optimization with tildev} corresponds to the frequency constraints on the control trajectory (in particular, to the function \( \ctrlconstraints \),) of the original optimal control problem, equivalently represented by the function  \( \Cfc \) in \(\R^{\dimProc}\). Recalling the definition \eqref{eq: ctrl freq constraints HD} of \( \Cfc \) we see that 
		\[
			\Cfc(\dummyst[0], \ldots, \dummyst[\hor], \dummyctrl[0], \ldots, \dummyctrl[\hor-1]) = \ctrlconstraints(\dummyctrl[0], \ldots, \dummyctrl[\hor-1]) \in \R^{\dimctrlconstraints}.
		\]
		Say \(\Cfc = \pmat{\Cfc_{1} & \cdots & \Cfc_{\dimctrlconstraints}} \transpose\). Its gradient is given by
		\[ 
			\derivative{\Cfc}{\proc} = \pmat{\derivative{\Cfc}{\dummyst[0]} & \cdots &  \derivative{\Cfc}{\dummyst[\hor]} & \derivative{\Cfc}{\dummyctrl[0]} & \cdots & \derivative{\Cfc}{\dummyctrl[\hor-1]}} \in \R^{\dimctrlconstraints \times \dimProc},
		\]
		where
		\begin{align*}
			 \derivative{\Cfc}{\dummyst[i]}
				 		& = \begin{pmatrix}
				 		\derivative{\Cfc_{1}}{\dummyst[i]^1} & \cdots &  \derivative{\Cfc_{1}}{\dummyst[i]^\dimst}\\  
				 		\vdots & \ddots  &\vdots \\
				 		\derivative{\Cfc_{\dimctrlconstraints}}{\dummyst[i]^1} & \cdots &  \derivative{\Cfc_{\dimctrlconstraints}}{\dummyst[i]^\dimst}  
				 		\end{pmatrix}\in \R^{\dimctrlconstraints\times\dimst}
				 		\text{ and } \\
				 	  \derivative{\Cfc}{\dummyctrl[i]}
				 		& = \begin{pmatrix}
				 		\derivative{\Cfc_{1}}{\dummyctrl[i]^1} & \cdots &  \derivative{\Cfc_{1}}{\dummyctrl[i]^\dimctrl}\\  
				 		\vdots & \ddots  &\vdots \\
				 		\derivative{\Cfc_{\dimctrlconstraints}}{\dummyctrl[i]^1} & \cdots &  \derivative{\Cfc_{\dimctrlconstraints}}{\dummyctrl[i]^\dimctrl}  
				 		\end{pmatrix}\in \R^{\dimctrlconstraints\times\dimctrl}.
			\end{align*}
			We see at once that \(\derivative{\Cfc}{\dummyst[i]}=0  \) for all \(i= 0, \ldots \hor\) and \(\derivative{\Cfc}{\dummyctrl[i]}= \derivative{\ctrlconstraints}{\dummyctrl[i]} \) for all \(i= 0, \ldots \hor-1,\). Therefore,
			\[ 
				\derivative{\Cfc}{\proc}
						=\begin{pmatrix}
							\mathbf{0}_{\dimctrlconstraints \times \dimst(\hor+1)} & 
							\derivative{\ctrlconstraints}{\dummyctrl[0]} & \cdots &  \derivative{\ctrlconstraints}{\dummyctrl[\hor-1]} 
						 \end{pmatrix} \in \R^{\dimctrlconstraints \times \dimProc}.
				\]
			Moreover, from the definition \ref{eq:sys:CFC} of \(\ctrlconstraints,\) \(\derivative{\ctrlconstraints}{\dummyctrl[i]}=\derivative{\cfc[i]}{\dummyctrl[]}\) for \(i=0,\ldots,\hor-1.\) This implies, for a vector \(\multipliercfc \in \R^{\dimctrlconstraints}\), 
			\[
			\biggl(\derivative{\Cfc}{\proc}\biggr)\transpose\multipliercfc=\begin{pmatrix}
			\mathbf{0}_{\dimst(\hor+1)\times 1}\\
			\Bigl(\derivative{\cfc[0]}{\dummyctrl[]}\Bigr)\transpose \multipliercfc  \\
			\vdots \\
			\Bigl(\derivative{\cfc[\hor-1]}{\dummyctrl[]}\Bigr)\transpose \multipliercfc			
			\end{pmatrix} \in \R^{\dimProc}.
			\]
			Therefore, the third term in \eqref{eq:necessary condition optimization with tildev} for the projection \eqref{eq:projection of dummy direction} of given \(\dummydir[]\) can be written as 
		\begin{equation}\label{eq:st projection ctrl freq constraints}
		\begin{aligned}
		\inprod{\biggl(\derivative{\Cfc}{\proc}(\optProc)\biggr)\transpose\multipliercfc}{\dummydir[]}=\sum\limits_{i=0}^{\hor}\inprod{0}{\stprojection[i]{(\dummydir[])}}+\sum\limits_{i=0}^{\hor-1}\inprod{\biggl(\derivative{\cfc[i]}{\dummyctrl[]}(\optCtrl)\biggr)\transpose \multipliercfc }{\ctrlprojection[i]{(\dummydir[])}}=0 .
		\end{aligned}
		\end{equation} 
	\item The fourth and fifth term in \eqref{eq:necessary condition optimization with tildev} correspond to the point-wise constraints on the states and the control actions of the original optimal control problem, respectively. For the equivalent optimization problem \ref{e:opt prob}, they are  multipliers corresponding to the state constraints such that each \( \MultiplierSt \in \Ncs_{\admstHD_{t}}(\optProc)\), \( \MultiplierCon \in \Ncs_{\admctrlHD_{t}}(\optProc), \) and from \eqref{eq:Multipliers in HD} we see that\\ 
	 \[\MultiplierSt = 
				\pmat{0, \ldots, 0, \multiplierstNcs, 0, \ldots, 0} \quad \text{ for } t = 0, \ldots, \hor,
	\]
	\[ 
		\MultiplierCon  =\pmat{0, \ldots, 0, \multiplierctrlNcs, 0, \ldots, 0} \quad \text{ for }t =0, \ldots, \hor-1.  
	\]
		Therefore, for the given \(\dummydir[]\) we can write
		\begin{equation*}
		        \begin{aligned}
		        \inprod{\MultiplierSt[s]}{\dummydir[]} =\begin{cases}
		        									     \inprod{\multiplierstNcs}{\dir_t^{\st[]}} &\text{for }  s=t,\\
				                                        0  &  \text{otherwise }, 
			                                           \end{cases} 	\hspace{1cm}
				\inprod{\MultiplierCon[s]}{\dummydir[]}=0  \text{ for } s= 0, \ldots, \hor-1,  
				\end{aligned}
		\end{equation*}
		The fourth and fifth terms in \eqref{eq:necessary condition optimization with tildev} become 
		\begin{equation}\label{eq: multiplier for pointwise st constraints}
		\inprod{\biggl(\sum\limits_{i=0}^{\hor}\MultiplierSt[i] + \sum_{i=0}^{\hor-1}\MultiplierCon[i] \biggr)}{\dummydir[]}
		= \inprod{\multiplierstNcs}{\dir_t^{\st[]}}.
		\end{equation}
		\item Finally, the term on the right hand side of \eqref{eq:necessary condition optimization with tildev} is a sum of the directional derivatives of the  functions \(\funInprod_{i}\) along the direction \(\dummydir[]\). The functions \(\funInprod_{i}\), defined in \eqref{eq:function innerprod of dynamics and adjoint}, correspond to the dynamics of the original system at the \(i^\text{th}\) time instance. 
		Recall the definitions \eqref{eq:directional derivative} and \eqref{eq:function innerprod of dynamics and adjoint} of the directional derivative and the functions \((\funInprod_{i})_{i=0}^{\hor-1}\), respectively, and observe that 
		\[
		\dirder[\tilde{\dir}]{\funInprod_i}(\optProc)= \lim\limits_{\theta \downarrow 0} \frac{\funInprod_i(\optProc + \theta \tilde{\dir}) - \funInprod_i(\optProc)}{\theta} \text{ and } \funInprod_i(\proc)=\inprod{\adjoint[i]}{\stprojection[i+1]({\proc})-\dynamics[i](\stprojection[i]{(\proc)},\ctrlprojection[i](\proc))}. 
		\] 
		Therefore the directional derivative of each \(\funInprod_i\)  along \(\dummydir[] \) is given by 
		\[
		\dirder[\tilde{\dir}]{\funInprod_i}(\optProc)=\lim\limits_{\theta \downarrow 0}\frac{\inprod{\adjoint[i]}{\theta\stprojection[i+1](\dummydir[])-\dynamics[i](\optState[i]+\theta\stprojection[i](\dummydir[]),\optCtrl[i]+\theta\ctrlprojection[i](\dummydir[]))+\dynamics[i](\optState[i],\optCtrl[i])}}
		{\theta}.
		\]
		% For the direction \(\dummydir[]\) as defined by the equation \eqref{eq:dummy dir state}, the state and control projections of \(\dummydir[]\) are given in the equation \eqref{eq:projection of dummy direction}.
		From the projections of \(\dummydir[]\) given in \eqref{eq:projection of dummy direction}, the above equation can be simplified to
		\[\dirder[\tilde{\dir}]{\funInprod_i}(\optProc)=\begin{dcases}
				                                       -\inprod{\adjoint}{\dirder[\dir^{\st[]}_{t}]{\dynamics(\cdot,\optCtrl)}(\optState)} & \text{ for } i=t,\\
				                                       \inprod{\adjoint[t-1]}{\dir^{\st[]}_{t}} & \text{ for } i=t-1,\\      
				                                        0&\text{otherwise.}                      
				                                        \end{dcases} 
		\]
	In other words, the term on the right hand side of \eqref{eq:necessary condition optimization with tildev} is given by
		\begin{align}
			 \sum_{i=0}^{\hor}\dirder[\tilde{\dir}]{\funInprod_i}(\optProc)
				 &=\begin{dcases}
				 		\dirder[\tilde{\dir}]{\funInprod_{0}}(\optProc) & \text{ for } t=0,\\
						\dirder[\tilde{\dir}]{\funInprod_{t-1}}(\optProc)
										 +\dirder[\tilde{\dir}]{\funInprod_t}(\optProc) & \text{ for } t=1,\ldots,\hor-1,\\
						\dirder[\tilde{\dir}]{\funInprod_{\hor-1}}(\optProc) & \text{ for } t=\hor,				
				 \end{dcases} \nonumber \\
				 \intertext{which leads to}
				\sum_{i=0}^{\hor} \dirder[\tilde{\dir}]{\funInprod_i}(\optProc)
				 &=\begin{dcases}
				 -\inprod{\adjoint[0]}{\dirder[\dir^{\st[]}_0]{\dynamics[0](\cdot,\optCtrl[0])}(\optState[0])} & \text{ for } t=0, \\
				 \inprod{\adjoint[t-1]}{\dir^{\st[]}_{t}}-\inprod{\adjoint}{\dirder[\dir^{\st[]}_{t}]{\dynamics(\cdot,\optCtrl)}(\optState)} 
				  & \text{ for } t=1,\ldots,\hor-1,
				 \\\inprod{\adjoint[\hor-1]}{\dir^{\st[]}_{\hor}}
				 & \text{ for } t=\hor.\label{eq: st projection of term corresponding to dynamics}
				 \end{dcases}
				 \end{align}
\end{itemize}
Substituting each term in \eqref{eq:necessary condition optimization with tildev} using equations \eqref{eq:st projection cost}, \eqref{eq:st projection st freq constraints}, \eqref{eq:st projection ctrl freq constraints}, \eqref{eq: multiplier for pointwise st constraints} and \eqref{eq: st projection of term corresponding to dynamics}, we arrive the following:
\begin{itemize}[label=\(\circ\), leftmargin=*]
 \item For \(t=0\),
	\begin{equation*}
	\inprod{-\multipliercost \derivative{\cost[0]}{\dummyst[]}(\optState[0],\optCtrl[0]) - \Bigl(\derivative{\sfc[0]}{\dummyst[]}(\optState[])\Bigr) \transpose \multipliersfc -  \multiplierstNcs[0]  }{\dir^{\st[]}_{0}}\leq -\inprod{\adjoint[0]}{\dirder[\dir^{\st[]}_{0}]{\dynamics[0](\cdot,\optCtrl[0])}(\optState[0])}.
	\end{equation*} 
 Since \( \dir^{\st[]}_{0} \in \R^\dimst \) is an arbitrary vector in \(\R^\dimst\), this means
\begin{align}
\inprod{-\multipliercost \derivative{\cost[0]}{\dummyst[]}(\optState[0],\optCtrl[0]) - \Bigl(\derivative{\sfc[0]}{\dummyst[]}(\optState[])\Bigr) \transpose \multipliersfc -  \multiplierstNcs[0]  }{\dirst}\leq -\inprod{\adjoint[0]}{\dirder[\dirst]{\dynamics[0](\cdot,\optCtrl[0])}(\optState[0])}
\end{align} 
for all \(\dirst \in \R^\dimst\).
\item For \(  t \in \set{1, \ldots, \hor-1}\) ,
\begin{equation*}
\begin{aligned}
\inprod{-\multipliercost \derivative{\cost}{\dummyst[]}(\optState,\optCtrl) - \Bigl(\derivative{\sfc}{\dummyst[]}(\optState[])\Bigr) \transpose \multipliersfc -  \multiplierstNcs  }{\dir^{\st[]}_{t}}\leq \inprod{\adjoint[t-1]}{\dir^{\st[]}_{t}} -\inprod{\adjoint}{\dirder[\dir^{\st[]}_{t}]{\dynamics(\cdot,\optCtrl)}(\optState)}.
\end{aligned}
\end{equation*}
Since \( \dir^{\st[]}_{t} \in \R^\dimst \) is an arbitrary vector in \(\R^\dimst\), this means
\begin{align}\label{eq:adjoint equation in proof}
\inprod{-\multipliercost \derivative{\cost}{\dummyst[]}(\optState,\optCtrl) - \Bigl(\derivative{\sfc}{\dummyst[]}(\optState[])\Bigr) \transpose \multipliersfc -  \multiplierstNcs  }{\dirst}\leq \inprod{\adjoint[t-1]}{\dirst} -\inprod{\adjoint}{\dirder[\dirst]{\dynamics(\cdot,\optCtrl)}(\optState)}.
\end{align}
\item For \(t=\hor\) 
\begin{equation*}
 \inprod{-\multipliercost \derivative{\cost[\hor]}{\dummyst[]}(\optState[\hor],\optCtrl[\hor])-\Bigl(\derivative{\sfc[\hor]}{\dummyst[]}(\optState[])\Bigr)\transpose \multipliersfc-\multiplierstNcs[\hor]}{\dir^{\st[]}_{\hor}}\leq  -\inprod{\adjoint[\hor-1]}{\dir^{\st[]}_{\hor}},
  \end{equation*}
  Since \( \dir^{\st[]}_{\hor} \in \R^\dimst \) is an arbitrary vector in \(\R^\dimst\), this means
 \begin{equation*}
  \inprod{-\multipliercost \derivative{\cost[\hor]}{\dummyst[]}(\optState[\hor],\optCtrl[\hor])-\Bigl(\derivative{\sfc[\hor]}{\dummyst[]}(\optState[])\Bigr)\transpose \multipliersfc-\multiplierstNcs[\hor]}{\dirst}\leq  -\inprod{\adjoint[\hor-1]}{\dirst} \text{ for all } \dirst \in \R^\dimst,
   \end{equation*} 
 which implies 
 \begin{equation} -\multipliercost \derivative{\cost[\hor]}{\dummyst[]}(\optState[\hor])-\Bigl(\derivative{\sfc[\hor]}{\dummyst[]}(\optState[])\Bigr)\transpose \multipliersfc-\multiplierstNcs[\hor]+\adjoint[\hor-1]=0.
 \end{equation}
  \end{itemize}
 Similarly, if we put  \(\dir=\hat{\dir}=(0,\cdots,0,\dir^{\ctrl[]}_t,0,\cdots,0)\in \R^\dimProc \) for some \(t\in\set{0,\cdots,\hor-1}\), in \eqref{eq:necessary condition for optimization problem}, then for all \(\dir_{t}^{\ctrl[]} \in \R^{\dimctrl}\), each term of \eqref{eq:necessary condition for optimization problem} equivalently translates to 

        \begin{equation}
        \begin{aligned}\label{eq:inner products}
        \begin{dcases}
       	\inprod{ \multipliercost \derivative{\Cost}{\proc}(\optProc)}{\hat{\dir}}
       	    	=\inprod{\multipliercost\derivative{\cost[t]}{\dummyctrl[]}(\optState[t],\optCtrl[t])}{\ctrlprojection[t]{(\hat{\dir})}}
       	=\inprod{\multipliercost\derivative{\cost[t]}{\dummyctrl[]}(\optState[t],\optCtrl[t])}{\dir^{\ctrl[]}_{t}}, \\
        \inprod{\Bigl(\derivative{\Sfc}{\proc}(\optProc)\Bigr)\transpose\multipliersfc}{\hat{\dir}}=\sum\limits_{i=0}^{\hor}\inprod{\Bigl(\derivative{\sfc[i]}{\dummyst[]}(\optCtrl)\Bigr)\transpose \multipliersfc }{\stprojection[i]{(\hat{\dir})}}+\sum\limits_{i=0}^{\hor-1}\inprod{0}{\ctrlprojection[i]{(\hat{\dir})}}=0, \\ 
        \inprod{\Bigl(\derivative{\Cfc}{\proc}(\optProc)\Bigr)\transpose\multipliercfc}{\hat{\dir}}
        	=\sum\limits_{i=0}^{\hor-1}\inprod{\Bigl(\derivative{\cfc[i]}{\dummyctrl[]}(\optCtrl[i])\Bigr)\transpose \multipliercfc}{\ctrlprojection[i]{(\hat{\dir})}} =\inprod{\Bigl(\derivative{\cfc[t]}{\dummyctrl[]}(\optCtrl[t])\Bigr)\transpose \multipliercfc}{\dir_{t}^{\ctrl[]}}, \\
        \inprod{\biggl(\sum\limits_{i=0}^{\hor}\MultiplierSt[i] + \sum_{i=0}^{\hor-1}\MultiplierCon[i] \biggr)}{\hat{\dir}}
				= \inprod{\multiplierctrlNcs}{\dir_t^{\ctrl[]}} ,
		\\ \dirder[\hat{\dir}]{\funInprod_t}(\optProc)= -\inprod{\adjoint}{\dirder[\dir^{\ctrl[]}_t]{\dynamics(\optState,\cdot)}(\optCtrl)}.
		\end{dcases}
		\end{aligned}
		\end{equation}
Therefore, the inequality \eqref{eq:necessary condition for optimization problem} can be rewritten as
\begin{equation}
\label{eq: Hamiltonian Maximization proof}
\begin{aligned}
	\inprod{ -\derivative{\cost}{\dummyctrl[]} \bigl( \optState, \optCtrl \bigr)-\derivative{\cfc}{\dummyctrl[]}(\optState,\optCtrl)- \multiplierctrlNcs}{\dir^{\ctrl[]}_t} \leq - \inprod{\adjoint}{\dirder[\dir^{\ctrl[]}_t]{\dynamics(\optState,\cdot)}(\optCtrl)} \\ \text{ for all } \dir^{\ctrl[]}_t\in \R^\dimctrl \text{ for } t=0,\ldots,\hor-1.
\end{aligned}
\end{equation}
Observe from \eqref{e:hamiltonian} that for \( t=0, \ldots, \hor-1,\) 

	for all \(\dirst \in \R^\dimst  \) we have,
		\begin{equation} \label{eq:dirder of Hamiltonian along in Rdimst}
			\begin{aligned}
			\dirder[\dirst]{\hamiltonian (\adjoint, t, \cdot, \optCtrl)}(\optState)=&\inprod{-\multipliercost \derivative{\cost}{\dummyst[]}(\optState,\optCtrl)  - \Big(\derivative{\sfc}{\dummyst[]}(\optState[])\Big) \transpose \multipliersfc}{\dirst}\\
 			& + \inprod{\adjoint}{\dirder[\dirst]{\dynamics(\cdot,\optCtrl)}(\optState)},
			\end{aligned}
		\end{equation}
and for all \(\dirctrl \in \R^{\dimctrl}\),
	\begin{equation}\label{eq:dirder of Hamiltonian along in Rdimctrl}
		\begin{aligned}
		\dirder[\dirctrl]{\hamiltonian (\adjoint, t, \optState, \cdot)}(\optCtrl)=&\inprod{-\multipliercost \derivative{\cost}{\dummyctrl[]}(\optState,\optCtrl)  - \
		\Bigl(\derivative{\cfc}{\dummyctrl[]}(\optCtrl[])\Bigr) \transpose \multipliercfc}{\dirctrl}\\
 		& + \inprod{\adjoint}{\dirder[\dirctrl]{\dynamics(\optState,\cdot)}(\optCtrl)}
		\end{aligned}.
\end{equation}

 Substituting \eqref{eq:dirder of Hamiltonian along in Rdimst} and \eqref{eq:dirder of Hamiltonian along in Rdimctrl} in \eqref{eq:adjoint equation in proof} and \eqref{eq: Hamiltonian Maximization proof} respectively, we get the adjoint equation and the Hamiltonian maximization condition in \ref{state and adjoint_equation}, \ref{Hamiltonian Maximization} in Theorem \ref{th:nonsmooth pmp}, respectively. Clearly, the state equations in \ref{state and adjoint_equation} follow from the definition of the Hamiltonian. Moreover, the non-negativity condition \ref{non-negativity} follows from \eqref{Opt_problem_necessary_condition}. Finally, the non-triviality condition \ref{non-triviality} follows from the non-triviality of \(\bigl(\multipliercost, \Multiplierfd, \multipliersfc, \multipliercfc \bigr) \) in \eqref{Opt_problem_necessary_condition}  and \(\Multiplierfd=(\adjoint[0],\ldots, \adjoint[\hor-1])\) in \eqref{Opt_problem_necessary_condition}.

	\section{Numerical Experiments}
	\label{sec:Numerics}
		%\begin{figure}[h!]%[H]
%	\centering
%%	\setlength\figureheight{=0.2\textwidth}
%%	\setlength\figurewidth{=0.35\textwidth}
%	%\includegraphics[scale=0.45]{Frequency_Profile}
%	\input{Control.tikz}
%	\caption{Optimal control profiles.}
%\label{fig:OptimalControl}
%\end{figure}
% % % % %Inverted Pendulum
\subsection{Example 1}
\label{Ex: Inverted Pendulum}
\textbf{Inverted Pendulum on a Cart System:} In our first example, we consider an optimal control problem with constraints on frequency components of control action for a linear discrete time model of pendulum on a cart system. First we design an optimal control for the problem via classical PMP and in order to satisfy constraints on frequency of control action we filter out the corresponding frequency components. Next, we use the necessary conditions proposed in Corollary \ref{cor:smooth sys PMP} to design an optimal control. The control actions obtained via both the methods are fed to the system and the corresponding system trajectories are observed.

%A discrete time model of a system is governed by the following difference equation:
%\begin{equation}\label{eq:DyPenCart}
%	\st[t+1]= A \st + B \ctrl,
%\end{equation} where \todo[color=green!20]{Either add the continuous time model or the discrete time model explicitly here.}
%\begin{itemize}[label=\( \circ\), leftmargin=*, align=left]
%	\item \( \R^{4} \ni \st \Let \pmat{\st \kth[1],\st \kth[2],\st \kth[3],\st \kth[4]}\) is a state vector with \(\st \kth[1]\) denoting the position of the cart, \(\st \kth[2]\) denotes the angle of the pendulum with respect to y-axis, \(\st \kth[3]\) the linear velocity of the cart, and \(\st \kth[4]\) denotes the angular velocity of the pendulum, at \(t^{\text{th}}\) time instant,
%	\item control \(\ctrl \in \R\) represents the force applied to the cart at \(t^{\text{th}}\) instant,
%\item and \( A = \pmat{ 1 &  -4.823*10^{-5} &  0.01 & -1.607*10^{-7} \\
%            0 & 1.001 & 0 & 1.0003*10^{-2} \\
%            0 & -9.649*10^{-3} &  1 & -4.8237*10^{-5} \\
%            0 & 1.994*10^{-1} & 0 & 1.001}\in \R^{4\times4},\ B=\pmat{1.9139\\ -3.2814 \\ 382.78 \\ -656.4}*10^{-5} \in \R^{4} \) are the system matrices.
%\end{itemize}
%
% % % % % % % % % % % % % % % % % % % % % % % % % % %55
The specifications of the system considered are given below:
\begin{table}[h!]
	\label{tab:table1}
	\centering
	\begin{tabular}{l|l}
		\hline
		\textbf{Parameter} & \textbf{Value} \\ \hline
		Mass of cart (\(M\)) & \(2.5\) \si{\kilo\gram}\\
		Mass of pendulum (\(m\)) & \(0.6\) \si{\kilo\gram}\\%
		Half-length of pendulum (\(l\)) & \(0.25\) \si{\centi\meter}\\%
		Range of cart track (\(-L\) to \(L\)) & \(L = 0.5\) \si{\centi\meter}\\%
		% Moment of inertia of pendulum \(J\) & \(\frac{4}{3}ml^{2}\) \si{\kilo\gram\square\centi\meter}\\%
		Acceleration due to gravity (\(g\) )& \(9.8\) \si{\meter \per \square \second}\\
		\hline
	\end{tabular}
	\caption{System Specifications}
\end{table}

The continuous time linearized model of the pendulum on a cart system is described by the following differential equation:
\begin{equation}
\label{ex1:continuous time system}
	\dot{\st[]} = A_{c} \st[] + B_{c} \ctrl[],
\end{equation}
where \( \st[] \in \R^{4}\),
\[
A_{c} = \begin{pmatrix}  0 & 0 & 1 & 0 \\ 0 & 0 & 0 & 1 \\ 0 & - \frac{(ml)^{2}g}{(J + ml^{2})(m + M - (m^{2}l^{2}/J+ ml^{2}))} & 0 & 0 \\ 0 & \frac{mgl(m + M)}{(J + ml^{2})(m + M - (m^{2}l^{2}/J+ ml^{2}))} & 0 & 0 \end{pmatrix} \quad
\text{and} \quad B_{c}= \pmat{0 \\ 0 \\ \frac{1}{(m + M - (m^{2}l^{2}/J+ ml^{2}))} \\ \frac{-ml}{(J + ml^{2})(m + M - (m^{2}l^{2}/J+ ml^{2}))}}. 
\]
We obtain the discretized model using a zero-order-hold technique; assuming constant input over the sampling time \(T_{s}=0.1 \si{\second}\) and the corresponding discrete time linearized systems dynamics is governed by the following difference equation:
\begin{equation}\label{eq:DyPenCart}
	\st[t+1]= A \st + B \ctrl,
\end{equation} where 
\begin{itemize}[label=\( \circ\), leftmargin=*, align=left]
	\item \( \R^{4} \ni \st \Let \pmat{\st \kth[1],\st \kth[2],\st \kth[3],\st \kth[4]}\) is the state vector with \(\st \kth[1]\) denoting the position of the cart, \(\st \kth[2]\) denotes the angle of the pendulum with respect to y-axis, \(\st \kth[3]\) the linear velocity of the cart, and \(\st \kth[4]\) denotes the angular velocity of the pendulum, at \(t^{\text{th}}\) time instant,
	\item control \(\ctrl \in \R\) represents the force applied to the cart at \(t^{\text{th}}\) instant,
\item and \( A \in \R^{4\times4},\ B \in \R^{4} \) are the system matrices.
\end{itemize}

\textbf{Problem Description:}
Our objective is to drive the cart from a given initial condition \(\initSt\) to a specified final position \(\finSt\) in \(\hor\) time steps; while minimizing the control effort and satisfying the following constraints
\begin{equation}\label{eq:constraintsPenCart}
\begin{cases}
	\lvert \st \kth[i] \rvert \leq \boundSt{i} \quad & \text{for } i = 1, \ldots, 4, \text{ and } t = 0, \ldots, \hor,\\
	\ctrl \leq \boundCtrl & \text{for } t = 0, \ldots, \hor-1,\\
	\hat{\ctrl[]}_{i}=0 & \text{for } i \not \in \mathcal{F},\\
	\st[0] = \initSt, \quad \st[\hor] = \finSt,
\end{cases}
\end{equation}
where the bounds \(\boundSt{i}, \boundCtrl\) on states and control represent pointwise constraints and the set of allowable frequencies \(\mathcal{F}\) represents frequency constraints. The optimal control problem is written as
\begin{equation}
	\begin{aligned}
		& \minimize_{\seq{\ctrl}{t}{0}{\hor-1}} && \sum_{t=0}^{\hor-1} \ctrl^{2}\\
		& \sbjto &&
		\begin{cases}
			\text{dynamics } \eqref{eq:DyPenCart}\text{ and constraints }\eqref{eq:constraintsPenCart}.
		\end{cases}
	\end{aligned}
\end{equation}
For this problem, we have assumed state bounds \((\boundSt{1}, \boundSt{2}, \boundSt{3}, \boundSt{4}) =  (0.2, \frac{20\pi}{180}, 15, 30)\), control bound \(\boundCtrl = 5 \si{\newton}\), length of horizon \( \hor= 240 \) and the set \( \mathcal{F}\Let \{1, \ldots, 96, 144, \ldots, 240\}\) corresponding to the low pass filter with cut off frequency \( \frac{4\pi}{5}\).

Using the first order necessary conditions proposed in Corollary \ref{cor:smooth sys PMP}, we obtain a two point boundary value problem. The solution of the boundary value problem is an optimal control, say \(\ctrl[\mathrm{p}]\) and denote its frequency components by \(\hat{\ctrl[]}_{\mathrm{p}}\). Consider an optimal control obtained through the classical discrete time PMP incorporating only pointwise constraints on control and state (i.e., neglecting the frequency constraints on the control profile).  The control so obtained is then passed through a proper filter in order to satisfy the frequency constraints on control profile, say \(u_{\mathrm{f}}\). Let \( \hat{u}_{\mathrm{f}}\) denote the frequency components of \(\ctrl[\mathrm{f}]\).
%we relax the frequency constraints on control and obtain a candidate optimal control with the help of classical discrete time PMP. In order to satisfy the constraints on the frequency components in \eqref{eq:constraintsPenCart} we filter out the undesired frequency components from this candidate optimal control. Let us name the control designed via this technique as an optimal control via filtering technique and denote the candidate optimal control obtained via filtering followed by classical PMP technique and its frequency components with \(\ctrl[f]\) and \(\hat{\ctrl[]}_{f}\) respectively. 

\begin{figure}[h!]
\centering
	\begin{subfigure}{0.5\textwidth}
	\centering
	\includegraphics[width=\textwidth]{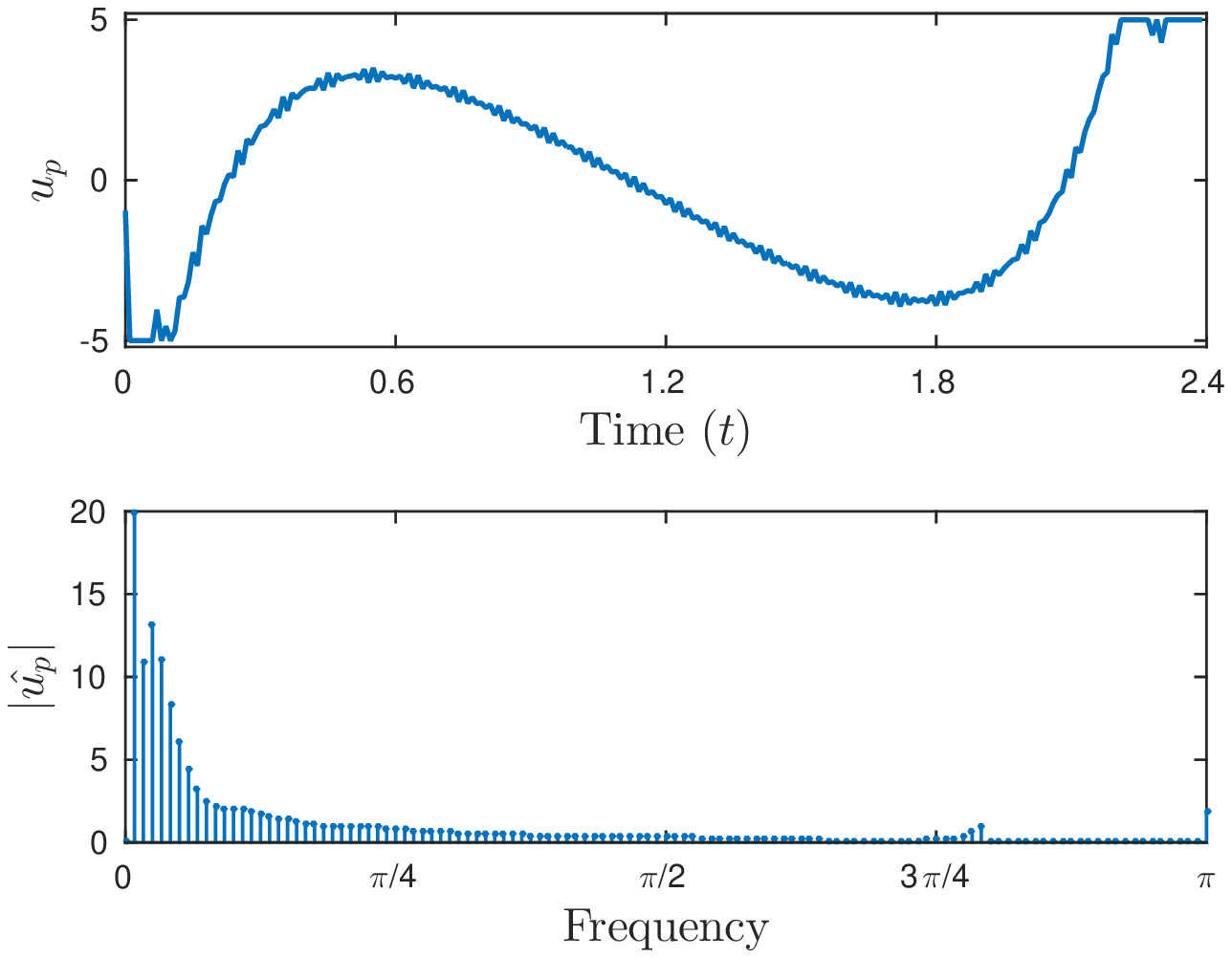}
	\caption{Control corresponding to frequency constrained PMP: Observe, \( \ctrl[\mathrm{p}] \) does not voialate the control bound}
    \label{fig:freqPMPCtrl}
	\end{subfigure} 
	~ \hspace{3mm}
	\begin{subfigure}{0.5\textwidth}
	\centering
	\includegraphics[width=\textwidth]{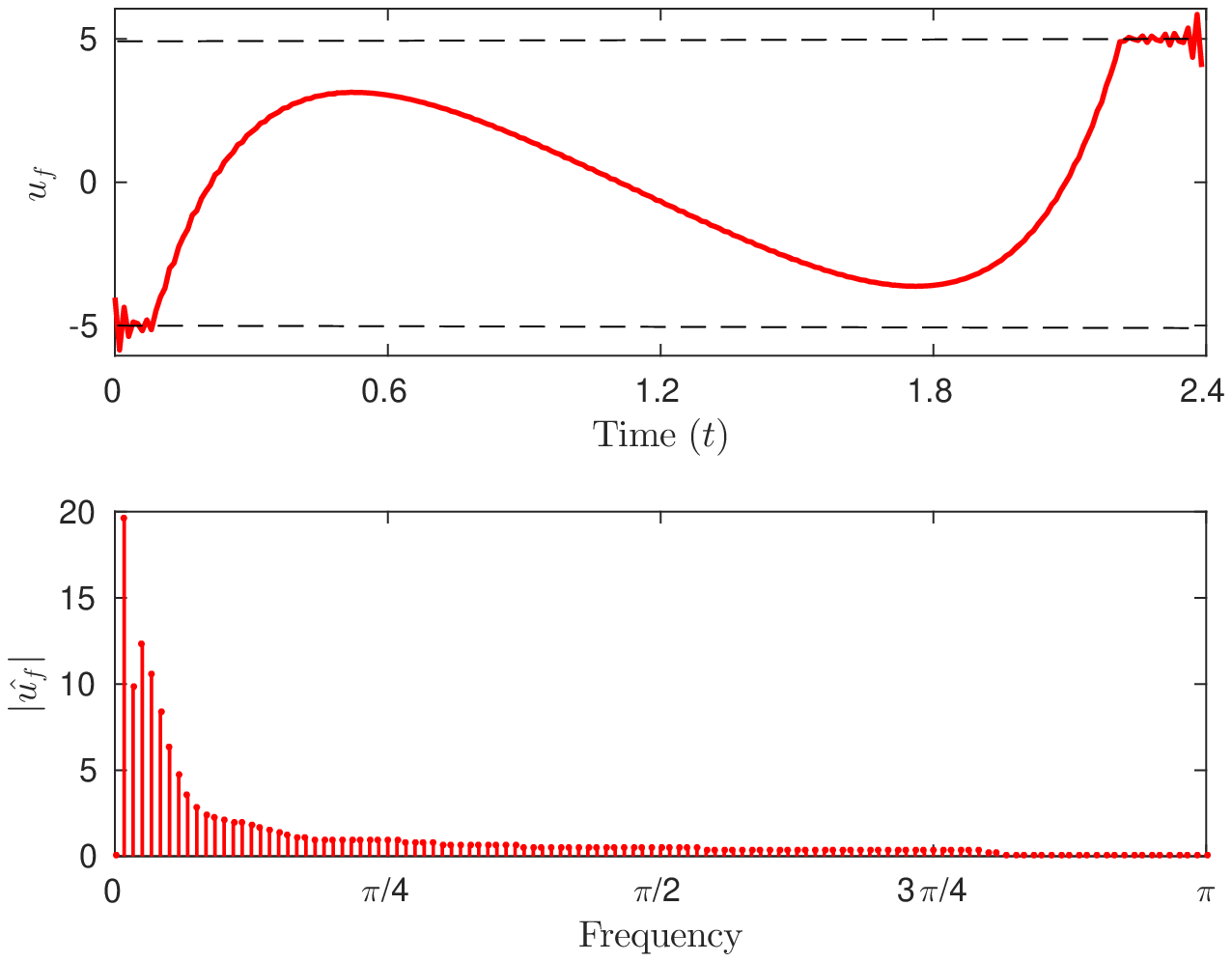}
	\caption{Filtered control: Observe, at the begining and at the end of the trajectory \( \ctrl[\mathrm{f}] \) violates the control bound}
	\label{fig:filtPMPCtrl}
	\end{subfigure}
	
		\begin{subfigure}{0.5\textwidth}
		\centering
		\includegraphics[width=\textwidth]{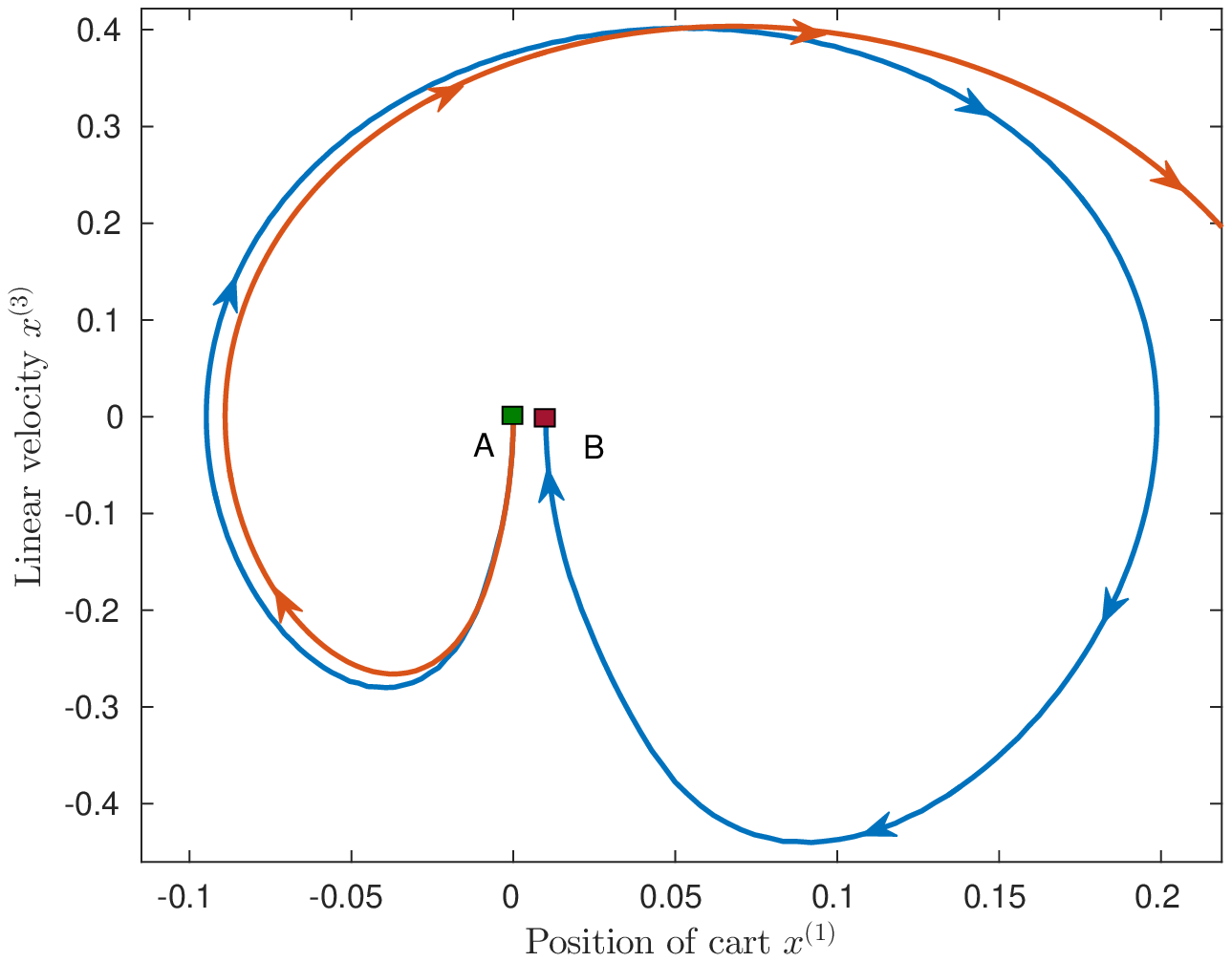}
		\caption{Phase portraits of the linear position of cart vs its linear velocity}
		\label{fig:Phase Portratit1}
	\end{subfigure}
	~ \hspace{3mm}
	\begin{subfigure}{0.5\textwidth}
		\centering
		\includegraphics[width=\textwidth]{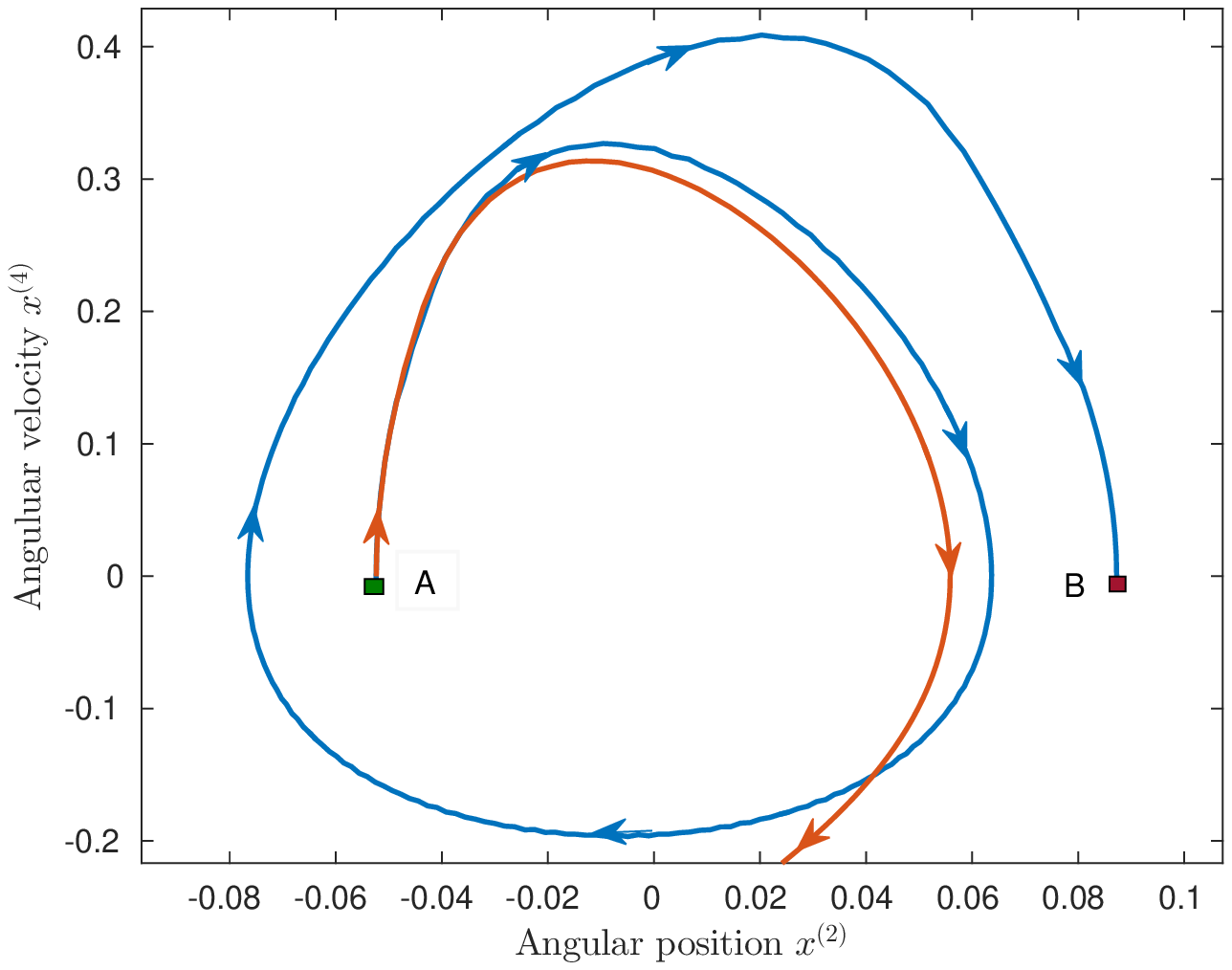}
		\caption{ Phase portraits of angular position of the pendulum vs its angular velocity}
		\label{fig:Phase Portratit2}
	\end{subfigure}
	\caption{Time and frequency domain profiles of the optimal control \(\ctrl[\mathrm{p}]\) (obtained using the proposed necessary conditions) (Fig. 3A), and \(\ctrl[\mathrm{f}]\) (the filtered optimal control) (Fig. 3B). Phase portraits of the system for the initial condition \(\initSt \) corresponding to the control inputs:  \( \ctrl[\mathrm{p}]\) 
	(in blue) and \(\ctrl[\mathrm{f}]\) (in red).} 
\label{fig:OptimalControl}
\end{figure}

The time domain and frequency domain profiles of both the control actions are shown in Figure \ref{fig:OptimalControl}. It is evident that the control \(\ctrl[\mathrm{p}]\) satisfies the constraints in the problem \ref{eq:constraintsPenCart}.

Both the controls \(\ctrl[\mathrm{p}]\) and \(\ctrl[\mathrm{f}]\) are fed to the system \eqref{eq:DyPenCart} and the  phase portraits corresponding to the initial condition \( \initSt \) are shown in Fig. 3C and Fig. 3D. It is clear from the plots that the control via frequency constrained PMP \(\ctrl[\mathrm{p}]\) respects pointwise state constraints and end point states are also attained. On the other hand the filtered control \(\ctrl[\mathrm{f}]\), because of removal of certain frequency components, is unable to maintain the pointwise state constraints and does not reach the final state.

Unlike the previous example, the next two examples incorporate non-smooth features in the system dynamics. In each case we present the necessary conditions of Theorem 4.1 for the particular case.
%\begin{figure}[h!]%[H]
%	% \setlength\figureheight{=0.2\textwidth}
%	% \setlength\figurewidth{=0.35\textwidth}
%	\includegraphics[scale=0.6, center]{STraject}
%	\caption{State trajectories}
%\label{fig:State Trajectories}
%\end{figure}
%\begin{figure}[h!]
%\centering
%	\begin{subfigure}{0.5\textwidth}
%	\centering
%	\includegraphics[width=\textwidth]{PhasePortrait1}
%	\caption{}%Phase portraits of linear position of cart vs linear velocity}
%    \label{fig:Phase Portratit1}
%	\end{subfigure}
%	~
%	\begin{subfigure}{0.5\textwidth}
%	\centering
%	\includegraphics[width=\textwidth]{PhasePortrait2}
%	\caption{} %Phase portraits of angular position of pendulum vs angular velocity}
%	\label{fig:Phase Portratit2}
%	\end{subfigure}
%\caption{Phase portraits of the system corresponding to the initial condition \(\initSt \) for the control input \( u_{p} \) and \( u_{f} \) \ref{fig:Phase Portratit1} Phase portraits of linear position of cart vs linear velocity and \ref{fig:Phase Portratit2} Phase portraits of angular position of pendulum vs angular velocity.}
%\label{fig:Phase Portraits}
%\end{figure}
% % % % % % % % % % % % % % % % % % % % % % % % % % % % % % % % % % % % % % % % % % % % % % % % % % % % % % %

		\subsection{Example 2}
\label{subsec:Ex-2}
In our second example, we consider a discrete time system \cite[Example 5.8, p.\ 98]{ref:GruPan-17}; governed by the following difference equation
\begin{equation}
\label{eq:Ex2 System}
	\st[t+1] = \dynamics[](\st, \ctrl)  \quad \text{ for } t=0,\ldots, \hor-1,
\end{equation}
where \(\st \in \R^{2}, \ctrl \in \R\), and \(\R^2 \times \R \ni (\xi, \mu) \mapsto \dynamics[](\xi, \mu) \Let \pmat{\xi \kth[1] (1 - \mu)\\ \norm{\xi} \mu} \in \R^{2}\).
Our objective is to characterize a solution of the following problem
\begin{equation}
\label{eq:Ex2 Opt Con Prob}
	\begin{aligned}
		& \minimize_{\seq{\ctrl}{t}{0}{\hor-1}} && \sum_{t=0}^{\hor-1} \inprod{\st}{\st} + \inprod{\ctrl}{\ctrl}\\
		& \sbjto &&
		\begin{cases}
			\text{dynamics \eqref{eq:Ex2 System}},\\
			\st[0] = \initSt,\\
			\ctrl \in [0,1] \quad \text{for } t = 0, \ldots, \hor-1.\\
			%\textcolor{blue}{\stconstraints \cdot (\st[0], \ldots, \st[\hor])^T = 0,} \\
			%\textcolor{blue}{\ctrlconstraints \cdot (\ctrl[0], \ldots, \ctrl[T-1])^T = 0}.
		\end{cases}
	\end{aligned}
\end{equation}
If \(\bigl(\seq{\optState}{t}{0}{\hor}, \seq{\optCtrl}{t}{0}{\hor-1}\bigr)\) is a solution of the problem \eqref{eq:Ex2 Opt Con Prob}, then Theorem \ref{th:nonsmooth pmp} gives the following first order necessary conditions:

there exist \(\multipliercost \in \set{0, 1}\) and a sequence of adjoint vectors \(\seq{\adjoint}{t}{0}{\hor-1}\) such that \( \bigl(\multipliercost, \seq{\adjoint}{t}{0}{\hor-1}\bigr) \neq 0\) and
\begin{itemize}[leftmargin=*]
	\item \textit{state and adjoint dynamics:}\\
	The function \( \dynamics[]\) in \eqref{eq:Ex2 System} is differentiable everywhere on \(\R^{2}\) except at a point \( (\st[] \kth[1], \st[] \kth[2])= (0, 0) \). Therefore, the adjoint dynamics on \(\R^{2} \setminus \{(0,0)\}\) is obtained from the condition \ref{smooth state and adjoint_equation} of the corollary \ref{cor:smooth sys PMP}  and the adjoint dynamics at \( (0,0) \) is obtained from the condition \ref{state and adjoint_equation} of Theorem \ref{th:nonsmooth pmp}.
	\begin{align}
		\label{eq:Ex2 State and adjoint dynamics}
		\text{for \(t = 0, \ldots, \hor-1,\)} \nonumber \\
		& \quad \optState[t+1] = \pmat{{\optState} \kth[1] (1 - \optCtrl)\\ \norm{\optState} \optCtrl},
		& \intertext{for \(t = 1, \ldots, \hor-1,\)} 
		\label{eq:Ex2 adjoint dynamics}
		& \begin{dcases}
			\adjoint[t-1]  = 2 \multipliercost \optState + \pmat{1 & 0\\ \dfrac{{\optState} \kth[1]}{\norm{\optState}} & \dfrac{{\optState} \kth[2]}{\norm{\optState}}} \transpose \adjoint \optCtrl \quad & \text{if } \optState \neq 0,\\
			\inprod{\adjoint[t-1]}{y} \geq \inprod{\adjoint}{\pmat{y \kth[1] (1-\optCtrl)\\ \norm{y} \optCtrl}} \quad \text{ for all } y \in \R^{\dimst} & \text{if } \optState = 0;
		  \end{dcases}
	\end{align}

	\item\textit{transversality:}\\
	\begin{equation}
	\label{eq:Ex2 transversality}
	\begin{dcases}
		2 \multipliercost \optState[0] + \pmat{1 & 0 \\ \dfrac{{\optState[0]} \kth[1]}{\norm{\optState[0]}} & \dfrac{{\optState[0]} \kth[2]}{\norm{\optState[0]}}} \transpose \adjoint[0] \optCtrl[0] = 0 & \text{if } \optState[0] \neq 0,\\
		\inprod{\adjoint[0]}{\pmat{y \kth[1] (1-\optCtrl[0])\\ \norm{y} \optCtrl[0]}} \leq 0 \quad \text{for all } y \in \R^{\dimst} \quad & \text{if } \optState[0] = 0;
	\end{dcases}
	\end{equation}

	\item \textit{Hamiltonian maximization:}\\
	The function governing dynamics of system \eqref{eq:Ex2 System} is smooth with respect to control variable \(\ctrl[]\). Therefore the Hamiltonian maximization condition obtained using the condition \ref{smooth Hamiltonian Maximization} of the Corollary \ref{cor:smooth sys PMP} is:
		\begin{equation}
		\label{eq:Ex2 H-Maximization}
			2 \multipliercost \optCtrl = \inprod{\adjoint}{\pmat{-1 \\ \norm{\optState}}}.
		\end{equation}
\end{itemize}

		\subsection{Example 3. Buck Converter:}
% \textbf{Working Principle:}
% \textcolor{magenta}{
% The operation of a DC-DC buck converter, shown in Figure \ref{fig:buck converter}, is controlled by a switch, \( S \), present in the circuit. If the switch \(S\) is closed then the circuit A is connected to the voltage source \(\inpVolt\), and if the switch is open then the circuit A is disconnected from the voltage source. When the switch \(S\) is closed, the current starts building up in the circuit, and soon after the switch \(S\) is opened the current starts decreasing. The opening and closing of the switch \( S \) is handled by a switching circuit. In the current controlled mode the switching circuit is  provided with the reference value of current \(\curRef\) and free running clock of period \( \ClkPeriod \). The inductor current is fed to the switching circuit,  which compares the inductor current \( \current[] \) with \(\curRef\), and once the current \(\current[] \) attains the value \( \curRef \) switching circuit opens the switch \( S \). The opened switch is then closed on the arrival of next clock at the input of switching circuit. For the details on the working of converter see \cite{Bimbhra}.
% }
%% % % % % % % % % % % % % % % % % % % % % % % % % % % % % % %

In \secref{subsec:Ex-2} we have considered a system whose dynamics is smooth with respect to the control and nonsmooth with respect to the states. In the current subsection we consider a more general example where the system dynamics is nonsmooth with respect to the states as well as the control. The following difference equation represents modified discrete time buck converter system \cite{Banarjee2000}. In particular, we  suppose the clock cycle and the reference current are given and consider the voltage as our control input.
\begin{equation}
\label{eq:buck dynamics}
\begin{aligned}
	\current[t+1]= \dynamics[](\current, \inpVolt)= 
	\begin{dcases}
		\exp \Bigl(\frac{-\Resistance \ClkPeriod}{\Inductance}\Bigr) \current + \bigg(1 - \exp \Bigl({\frac{-\Resistance \ClkPeriod}{\Inductance}}\Bigr)\bigg) \frac{\inpVolt}{\Resistance} \quad & \text{if } \current \leq \curBor (\inpVolt),\\
		\Big(\frac{\inpVolt - \current \Resistance}{\inpVolt - \curRef \Resistance}\Big)^{\bigl(\frac{\Resistance + \DiodeRes}{\Resistance}\bigr)} \curRef \exp {\Bigl(-\frac{\Resistance + \DiodeRes}{\Inductance} \ClkPeriod \Bigr)} &  \text{if } \current \geq \curBor (\inpVolt),
	\end{dcases}
\end{aligned}
\end{equation}
where  the different physical variables and parameters are as follows:
\begin{itemize}[leftmargin=*, label=\(\circ\), align=left]
	\item  \(\ClkPeriod\) is the time period of clock pulse,
	\item \(\current\) is the inductor current at \(t^{th}\) instant of time, and \(\inpVolt\) is the control input voltage,
	\item  \(\Resistance\) is the load resistance, \(\Inductance\) is the inductance of inductor, and  \(\DiodeRes\) is the  diode resistance,
	\item \(\curRef\) is the reference current, \(\curBor (\inpVolt[]) \Let \Big(\curRef - \dfrac{\inpVolt[]}{\Resistance}\Big) \exp \Big(\dfrac{\Resistance  \ClkPeriod}{\Inductance}\Big) + \dfrac{\inpVolt[]}{\Resistance}\) is the borderline current at input voltage \(\inpVolt[]\), and \(\vltBor (\current[]) \Let \Resistance \Big(1 - \exp \Big(\dfrac{\Resistance \ClkPeriod}{\Inductance} \Big) \Big)^{-1} \Big(\current[] - \curRef \exp \Big(\dfrac{\Resistance \ClkPeriod}{\Inductance}\Big)\Big)\) is the borderline voltagei when the current is \(\current[]\).
\end{itemize}

For more details on the dynamics, the reader may refer to \cite{Banarjee2000}. To further compress the notation, we define
\begin{itemize}[leftmargin=*, label=\(\circ\), align=left]
	\item \(\sysParA \Let \exp\Big(\dfrac{-\Resistance \ClkPeriod}{\Inductance}\Big)\), \(\sysParB \Let \dfrac{1 - \sysParA}{\Resistance}\), \(\sysParC \Let \Resistance \curRef \exp \Big(-\dfrac{\big(\Resistance + \DiodeRes\big)}{\Inductance}\ClkPeriod \Big)\),
	\item \(\sysParD \Let \curRef \exp \Big(-\dfrac{(\Resistance + \DiodeRes)}{\Inductance}\ClkPeriod\Big)\), \( \refVlt = \curRef \Resistance\)
\end{itemize}

Assuming that the diode resistance is very small compared to the load resistance, \(\frac{\Resistance + \DiodeRes}{\Resistance} \approx 1\), the dynamics is represented in a compact form as
\begin{equation}
\label{eq:buck dynamics_simpl}
\begin{aligned}
	\current[t+1]= \dynamics[](\current, \inpVolt) = \begin{dcases}
														\sysParA \current + \sysParB \inpVolt \quad & \text{if } \current \leq \curBor (\inpVolt),\\
														- \sysParC \Big(\frac{\current}{\inpVolt - \refVlt}\Big) + \sysParD \Big(\frac{\inpVolt}{\inpVolt - \refVlt}\Big) &  \text{if } \current \geq \curBor (\inpVolt).
	\end{dcases}
\end{aligned}
\end{equation}
Observe that the map \(\R \times \R \ni (\current[], \inpVolt[]) \mapsto \dynamics[](\current[], \inpVolt[]) \in \R\) in \eqref{eq:buck dynamics_simpl} governing the dynamics of the system is nonsmooth at \(\current[] = \curBor (\inpVolt[]).\)

% \textcolor{magenta}{The optimal control problem is defined for a fixed initial value of the current \((\current[0])\) and 
% a fixed final value of the current \( (\current[\hor]) \), and minimizes the power loss in the load, while satisfying the state and control constraints at every \( t \in 0,\ldots, \hor \). The precise mathematical statement is as given below:}

We consider the following optimal control problem where we minimize the power loss in the load and input voltage while transferring the system from a given initial state \(\current[0] = \bar{i}\) to a given final state \(\current[\hor] = \current[\mathrm{f}]\) in \( \hor\) time steps.   
\begin{equation}
\label{eq:opti_cntrl_prb_buck_converter}
\begin{aligned}
	& \minimize && \sum\limits_{t=0}^{\hor-1} \current^{2}+\inpVolt^{2}\\
			& \sbjto && \begin{cases}
				\text{dynamics } \eqref{eq:buck dynamics_simpl}\\
				% \current \in \admst  \quad \text{for } t = 1, \ldots, \hor-1,\\
				% \inpVolt \in \admctrl[],\\
				\current[0] = \current[i],\\
				\current[\hor] = \current[f].
			\end{cases}
\end{aligned}
\end{equation}

If \((\seq{\optCur}{t}{0}{\hor}, \seq{\optVlt}{t}{0}{\hor-1})\) is a solution of the problem \eqref{eq:opti_cntrl_prb_buck_converter}, then the first order necessary conditions in Theorem \ref{th:nonsmooth pmp} translate to the following conditions: there exist \(\multipliercost \in \set{0, 1}\) and a sequence \((\adjoint)_{t=0}^{\hor-1}\), not all zero simultaneously, satisfying
\begin{itemize}[leftmargin=*, label=\(\circ\), align=left]
	\item \emph{state and adjoint dynamics:}
		\begin{align}
		\label{eq:EX3 State and adjoint dynamics}
			\intertext{for \(t = 0, \ldots, \hor-1,\)} \nonumber \\
			& \optCur[t+1] = 
			  \begin{dcases}
				\sysParA \optCur + \sysParB \optVlt \quad & \text{if } \optCur < \curBor (\optVlt),\\
				- \sysParC \Big(\frac{\optCur}{\optVlt - \refVlt}\Big) + \sysParD  \Big(\frac{\optVlt}{\optVlt - \refVlt}\Big) &  \text{if } \optCur \geq \curBor (\optVlt),
			  \end{dcases}\\
			% Adjoint Dynamics
			\intertext{for \(t = 1, \ldots, \hor-1\),} \nonumber \\
			\label{eq:EX3 adjoint dynamics}
			& \begin{dcases}
				\adjoint[t-1] = -2 \multipliercost \optCur - \multiplierstNcs + \sysParA \adjoint \quad  & \text{if } \optCur < \curBor,\\
				\adjoint[t-1] = -2 \multipliercost \optCur - \multiplierstNcs - \frac{\sysParC}{\optVlt - \refVlt} \adjoint & \text{if } \optCur > \curBor (\optVlt),\\
				- \frac{\refVlt \adjoint}{\optVlt - \refVlt} \leq \frac{\adjoint[t-1] + \multiplierstNcs + 2 \multipliercost \optCur}{\sysParA} \leq \adjoint & \text{if } \optCur = \curBor (\optVlt);
			  \end{dcases}
		\end{align}
		Notice that when \(\optCur \neq \curBor (\optVlt)\), the dynamics is smooth with respect to the states and the adjoint dynamics is given by the classical discrete time PMP and when \(\optCur = \curBor (\optVlt)\), the adjoint dynamics is given by the set inclusion in Theorem \ref{th:nonsmooth pmp}.

	\item \emph{transversality}
		\begin{equation}
		\label{eq:EX3 transversality}
		\begin{aligned}
		& \begin{dcases}
			2 \multipliercost \optCur[0] + \multiplierstNcs[0] = \sysParA \adjoint[0] \quad & \text{if } \optCur < \curBor (\optVlt),\\
			2 \multipliercost \optCur[0] + \multiplierstNcs[0] = - \sysParA \frac{\refVlt}{\optVlt - \refVlt} \adjoint[0] & \text{if } \optCur > \curBor (\optVlt),\\
			- \frac{\refVlt \adjoint[0]}{\optVlt[0] - \refVlt} \leq \frac{\multiplierstNcs[0] + 2 \multipliercost \optCur[0]}{\sysParA} \leq \adjoint[0] & \text{if } \optCur = \curBor (\optVlt),
		\end{dcases} \quad \text{and}\\
			& \quad \adjoint[\hor] = - \multiplierstNcs[\hor];
 		\end{aligned}
 		\end{equation}
	\item \emph{Hamiltonian maximization}
		\begin{equation}
		\label{eq:EX3 H maximization}
		\begin{aligned}
		\begin{dcases}
			\multipliercost \optVlt = \frac{ \sysParC \optCur + \sysParD \refVlt}{2 (\optVlt - \refVlt)^{2}} \adjoint \quad & \text{if } \optVlt < \vltBor (\optCur),\\
			\multipliercost \optVlt = \frac{\sysParB}{2} \adjoint & \text{if } \optVlt > \vltBor (\optCur),\\
			\sysParB  \adjoint \leq 2 \multipliercost \optVlt \leq \frac{\sysParC \optCur + \sysParD \refVlt}{(\optVlt - \refVlt)^{2}} \adjoint & \text{if } \optVlt = \vltBor (\optCur).
		\end{dcases}
		\end{aligned}
		\end{equation}
		The dynamics is smooth with respect to the control when \(\optVlt \neq \vltBor (\optCur)\) and hence the Hamiltonian maximization condition is same as the one given by Corollary \ref{cor:smooth sys PMP}. If \(\optVlt = \vltBor (\optCur)\), then the system dynamics is not differentiable with respect to the control and the Hamiltonian maximization condition corresponds to a set inclusion.
\end{itemize}

The adjoint equation \eqref{eq:Ex2 adjoint dynamics} and the transversality condition \eqref{eq:Ex2 transversality} for the problem \eqref{eq:Ex2 Opt Con Prob}, and the adjoint equation  \eqref{eq:EX3 adjoint dynamics}, the transversality condition \eqref{eq:EX3 transversality} and the Hamiltonian maximization condition \eqref{eq:EX3 H maximization} for the problem \eqref{eq:opti_cntrl_prb_buck_converter} involve set-theoretic inclusion conditions. Most commercially available algorithms for solving two point boundary value problems involve equalities instead of inclusions, and can, therefore, not be applied directly to these problems. However, techniques based on semismooth Newtons methods (see, e.g., \cite{ref:Hin-10}) appear to be promising directions for our problems; numerical algorithms based on these methods for synthesis of optimal control trajectories are under investigation and will be reported subsequently.

	\appendix
		\section{Nonsmooth Optimization and Convex Cones}
\label{Appendix A}
\begin{theorem} \cite[Theorem 10.47 on p.\ 221]{Clarke-13}
	\label{Clarke's multiplier rule}
	Consider an optimization problem
	\begin{equation}
	\label{gen opti problem}
	\begin{aligned}
		& \minimize && \Cost(\proc)\\
		& \sbjto && \begin{cases}
			\equalconstrs(\proc) = 0,\\
			\inequalconstrs(\proc) \leq 0,\\
			\proc \in \genSet,\\
		\end{cases}		
	\end{aligned}
	\end{equation}
	where the functions governing cost, equality constraints, and inequality constraints are given by the maps \(\R^{\genDim} \ni \proc \mapsto \Cost(\proc)\in \R\), \(\R^{\genDim} \ni \proc \mapsto \equalconstrs(\proc)\in \R^{\dimequal}\), and \( \R^{\genDim} \ni \proc \mapsto \inequalconstrs(\proc) \in \R^{\diminequal}\) respectively, and \(\genSet\) is a closed subset of \(\R^{\genDim}\). If \(\optimal{z}\) solves \eqref{gen opti 	problem} and \(\Cost,\equalconstrs,\inequalconstrs\) are Lipschitz near \(\optimal{z}\), then there exist \((\multipliercost, \multiplierequal, \multiplierinequal) \in \R \times \R^{\dimequal} \times \R^{\diminequal}\) satisfying
	\begin{enumerate}[label={\textup{(\roman*)}}, leftmargin=*, align=left, widest=iii]
		\item \label{Clarke:non-triviality} the nontriviality condition
			\[
				(\multipliercost,\multiplierequal,\multiplierinequal) \neq 0;
			\]

		\item \label{Clarke:non-negativity and complementary slackness} non-negativity and complementary slackness
			\[
				\multipliercost \in \set{0, 1}, \quad \multiplierinequal \geq 0, \quad \inprod{\multiplierinequal}{\inequalconstrs(\optimal{z})} = 0;
			\]

		\item \label{Clarke:stationary condition} and the stationarity condition
			\[
				0 \in \ggrad{\Bigl(\multipliercost \Cost +\inprod{\multiplierequal}{\equalconstrs} +\inprod{\multiplierinequal}{\inequalconstrs}\Bigr)}{\optimal{z}} +\Ncs_{\genSet}(\optimal{z}).
			\]
	\end{enumerate}
\end{theorem}

\begin{theorem}\cite[Theorem 3 on p.\ 7]{ref:Bol-75}
	\label{th:closure of convex hull}
	Let \( K_{1}, \ldots, K_{s}\)	 be closed convex cones in \(\R^{n}\) with vertex at \(0\). If the cone \(K = \chull \bigl(\bigcup_{i=1}^{s} K_{i} \bigr)\) is not closed, then there are vectors \(\lambda_{1} \in K_{1}, \ldots, \lambda_{s} \in K_{s}\), not all of them zero, such that \(\sum_{i=1}^{s} \lambda_{i} = 0\).
\end{theorem}
\begin{theorem}\cite[Theorem 4 on p.\ 8]{ref:Bol-75}
	\label{th:Dual cone of intersection is union of dual cone} Let \( K_1, \ldots, K_s \) be closed convex cones in \(\R^\dimst\) with vertex at \(\st[0]\). Then \(\dualCone{ \left(\bigcap_{i=1}^{s} K_i\right)}=\closure \left(\chull \left( \bigcup_{i=0}^{s} \dualCone{K_i} \right)\right)\).

\end{theorem}

\section{Auxiliary Lemmas}
\label{sec:Lemmas}

\begin{lemma}\cite[Chapter 10 on p.\ 201]{Clarke-13}
\label{ggrad_of Lipschitz_+_C1}
	Consider two functions \(\function_{1}, \function_{2}: \R^{\dimProc} \lra \R\) such that in an open neighbourhood of \(\proc \in \R^{\dimProc}\), \(\function_{1}\) is Lipschitz continuous and \(\function_{2}\) is continuously differentiable. If  \(\ggrad{\function}{\cdot}\) denotes the generalized differential then,
	\[
		\ggrad{(\function_{1} + \function_2)}{\proc} = \ggrad{\function_{1}}{\proc} + \derivative{\function_2}{\proc} (\proc).
	\]
\end{lemma}
% % % % % % % % % % % % % % % % % % % % % %
\begin{lemma}
\label{Tcs to intersection is intersection of Tcs}
	Let \(N\) be a positive integer,  \(\genSet_{i} \subset \R^{\dimctrl} \) be closed and nonempty sets for \(i = 1, \ldots, N\). Then \(\Tcs_{\genSet}(\optProc) = \bigcap_{i=1}^{N} \Tcs_{\genSet_{i}}(\optProc)\) for  \(\optProc \in \genSet \Let \bigcap_{i=1}^{N} \genSet_{i} \),
\end{lemma}

\begin{proof}
	For \(\dir \in \Tcs_{\genSet}(\optProc)\), \(\gdd[\dist{\genSet}]{\optProc}{\dir}= 0\). Since \( \genSet \subset \genSet_{i}\), \(\gdd[\dist{\genSet_{i}}]{\optProc}{\dir} = 0\) for each \(i = 1, \ldots, N\) and hence \(\dir \in \bigl(\bigcap_{i=1}^{N} \Tcs_{\genSet_{i}}(\optProc)\bigr).\) Conversely, if \(\dir \in \bigcap_{i=1}^{N} \Tcs_{\genSet_{i}}(\optProc)\), then  \(\gdd[\dist{\genSet_{i}}]{\optProc}{\dir}= 0 \text{ for } i= 1, \ldots, N\). Thus, \(\gdd[\dist{\genSet}]{\optProc}{\dir} = 0\) and \(\dir \in \Tcs_{\genSet}(\optProc)\).
\end{proof}

		\label{Appendix}
	% References
	\bibliographystyle{amsalpha}
	 \bibliography{references}

\end{document}